\documentclass[11pt, epsfig]{article}
\usepackage{epsfig, amsmath, amssymb, amsthm, times}
\usepackage{graphicx}
\usepackage{color}
\usepackage{bm} 
\usepackage{stackrel}

\input{macros.sty}

\def\be{\begin{equation}}
	\def\ee{\end{equation}}

\def\1{\mathbf{1}}

\newcommand{\cyc}[1]{\operatorname{Cyc}\left(#1\right)}

\DeclareMathSymbol{\dash}{\mathord}{AMSa}{"39}
\newcommand{\Glim}[1]{\Gamma\dash\lim_{#1}}

\newtheorem{theorem}{Theorem}[section]

\newtheorem{lemma}[theorem]{Lemma}
\newtheorem{cor}[theorem]{Corollary}
\newtheorem{remark}[theorem]{Remark}
\newtheorem{definition}[theorem]{Definition}

\title{{The Kuramoto model on the Sierpinski gasket II: Twisted states}}

\author{Georgi S. Medvedev \and Matthew S. Mizuhara}

\begin{document}
	
	\maketitle
	
	\begin{abstract}
		We study the Kuramoto model (KM) of coupled phase oscillators on graphs approximating the Sierpinski gasket (SG).
		As the size of the graph tends to infinity, the limit points of the sequence of stable equilibria in the KM correspond
		to the minima of the Dirichlet energy, i.e., to harmonic maps from the SG to the circle. We provide a complete
		description of the stable equilibria of the continuum limit of the KM on graphs approximating the SG,
		under both Dirichlet and free boundary conditions. We show that there is a unique stable equilibrium in each homotopy
		class of continuous functions from the SG to the circle. These equilibria serve as generalizations of the classical twisted states on ring networks. Furthermore, we extend the analysis to the KM on
		post-critically finite fractals. The results of this work reveal the link between self-similar organization and
		network dynamics.
	\end{abstract}

	\leftline{\small 2020 {\it Mathematical Subject Classification.\/}
		34C15, 
		28A80, 
		37B35, 
		37M15, 
		58E30, 
		92B20.           
	}
	
	
	\section{Introduction}\label{sec.intro}
	\setcounter{equation}{0}


	In this paper, we initiate the study of the Kuramoto model (KM) of coupled phase oscillators on
	self-similar networks. This model has been extensively studied in nonlinear science \cite{Str00} and in statistical physics, where it is known under the name of 
	$XY$-model \cite{CosSha2021}.
	In particular, it has been used for understanding the role of the network topology on emergent
	dynamics and synchronization \cite{CMM18}.
	In a similar vein, we use the KM to explore how self-similar network organization influences
	network dynamics. There is empirical evidence of self-similarity in natural and technological
	networks \cite{Mandelbrot-Fractal-Geometry}.
	Often, this is a by-product of a network’s hierarchical structure, as seen
	in the Internet, where smaller modules mirror the structure of the larger network
	\cite{DillKum02}.
	While self-similarity in real-life networks is not exact from a mathematical standpoint,
	the analysis of idealized models on fractals
	can offer insights into the implications of self-similarity observed in such networks.

	As a model of self-similar connectivity, we use a sequence of graphs approximating a
	fractal domain $K\subset\R^d$. To simplify presentation, we first focus on the case
	when $K$ is a Sierpinski Gasket (SG), a canonical example of a fractal set \cite{Kig01, Strich2006}.
	In Section~\ref{sec.pcf}, we will extend this analysis to the KM on graphs approximating
	a post-critically finite (p.c.f.) fractal. 
	
	Until Section \ref{sec.pcf}, $K\subset\R^2$ is the unique nonempty compact set
	satisfying the following fixed point equation
	\begin{equation}\label{fixed-point}
		K=\bigcup_{i=1}^3 F_i(K),\qquad F_i(x)=2^{-1}(x-v_i)+v_i,
	\end{equation}
	where $v_1, v_2, v_3$ are three distinct points in the plane (see Fig.~\ref{fig:sg}).
	Without loss of generality, we set $v_1=(0,0),$ $v_2=(1/2, \sqrt{3}/2)$, and
	$v_3=(1,0)$, so that $v_1, v_2,$ and $v_3$ are vertices of an equilateral triangle
	with side length $1$.
	
	\begin{figure}[h]
		\centering
		\includegraphics[width = .35\textwidth]{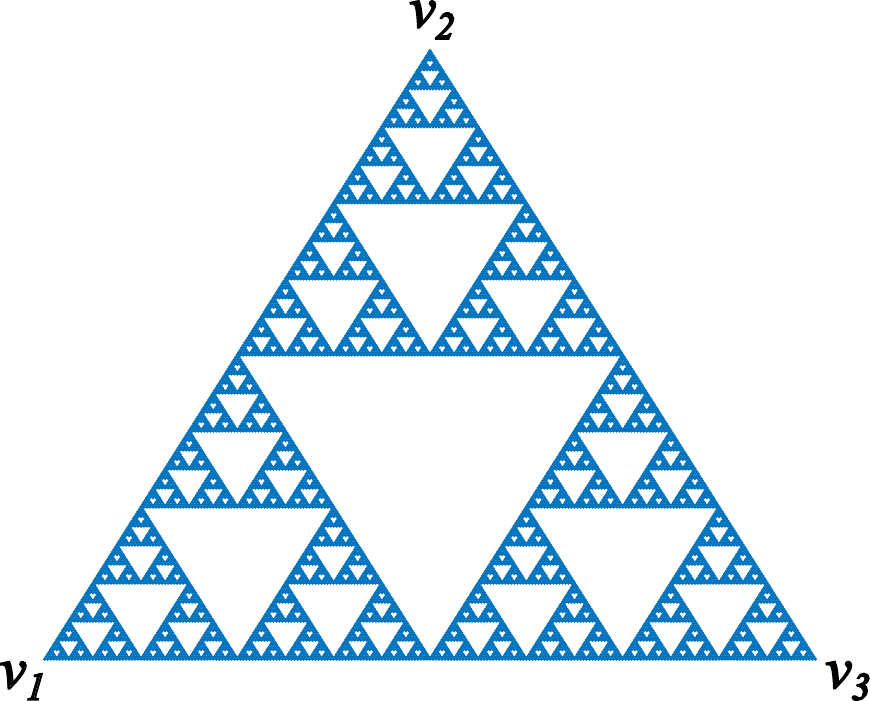}
		\caption{SG.}
		\label{fig:sg}
	\end{figure}

	The sequence of graphs $\Gamma_n=(V_n,E_n)$ approximating $K$ is constructed as follows.
	Let $V_n$ stand for the set of vertices of $\Gamma_n$. Then $V_0=\{ v_1, v_2, v_3\}$ and
	for $n\ge 1$ set
	$$
	V_n=\bigcup_{i=1}^3 F_i(V_{n-1}).
	$$
	Further, $\Gamma_0$ is the complete graph on three nodes. To describe the set of edges of $\Gamma_n$,
	we need an alphabet $S\doteq \{1,2,3\}$
	and the set of words consisting of $n$ symbols, $S^n$.
	Then  $v_i,v_j \in V_n$ are adjacent (denoted by either $v_i\sim_n v_j$, or simply $i\sim_n j$) if there is $w=(w_1,w_2,\dots, w_n)\in S^n$,
	such that 
	$$
	v_i,v_j\in F_w(V_0), \quad\mbox{where}\quad F_w\doteq F_{w_1}\circ F_{w_2}\circ\dots\circ F_{w_n}.
	$$
	Geometrically, $\Gamma_0$ is a triangle with vertices $v_1, v_2,$ and $v_3$, and
	$\Gamma_n=\bigcup_{w\in S^n} F_w(\Gamma_0)$. Two nodes $v_i,v_j\in V_n$ are adjacent
	if both belong to the same $n$--cell $F_w(\Gamma_0)$ for some $w\in S^n$
	(see Fig.~\ref{fig:sg_graphs}).
	\begin{figure}[h]
		\centering
		\includegraphics[width = .9\textwidth]{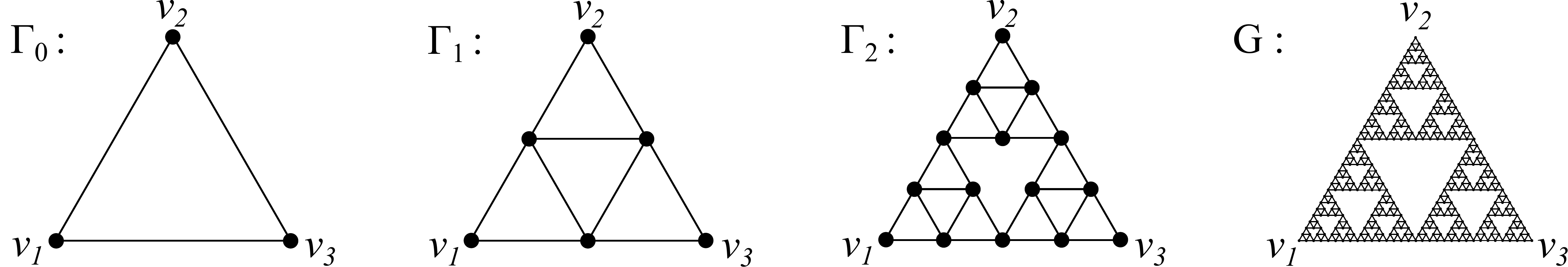}
		\caption{Graphs $\Gamma_n$ approximating SG.}
		\label{fig:sg_graphs}
	\end{figure}
	There are $\frac{3}{2}\left(3^n-1\right)$ vertices of $\Gamma_n$. We will distinguish between
	the boundary vertices given by $V_0=\{v_1, v_2, v_3\}\subset V_n$ and the interior vertices
	$V_n\setminus V_0$. Note that the interior vertices have degree $4$ while the boundary ones
	have degree $2$.
	
	We pause to explain the  coding of the nodes of the graphs
	approximating $K$. Let $T$ denote the solid triangle with vertices $v_1, v_2,$ and
	$v_3$. Every node from $V_n$ is a vertex of the corresponding triangle
	$$
	T_w=F_w(T),\; w=(w_1,w_2,\dots ,w_n) \in S^n
	$$ 
	and can be represented
	as
	$$
	\bigcap_{k=1}^\infty
	F_{ w\underbrace{ \scriptscriptstyle iii\dots i}_\text{\normalfont $k$ times}}\,(T).
	$$
	For each node in $V_n$ we define an itinerary $w\bar{i}$,
	where $\bar{i}$ stands for the infinite sequence of $i$'s: $iii\dots,$ $i\in S$.
	Note that for $v\in V_n\setminus V_0$ there are two possible itineraries, e.g.,
	$v_{1\bar{2}}$ and $v_{2\bar{1}}$
	correspond to the same node from $V_2$.   We order the vertices in $V_n$ by the
	corresponding itineraries in lexicographical order:
	\begin{equation}\label{order-nodes}
		V_n=\{ v_{\bar 1}, v_{\bar 2}, v_{\bar 3}, v_{1\bar 1}, v_{1 \bar 2}, \dots\}.
	\end{equation}
	Abusing notation, we will use $v_1, v_2, v_3, v_4, \dots$ to denote the elements of $V_n$
	appearing in the same order as in \eqref{order-nodes}.

	Now we are ready to formulate the KM on $\Gamma_n$:
	\begin{equation}\label{KM-intro}
		\dot u(t,v_i)=\left(\frac{5}{3}\right)^n \sum_{j:j\sim_n i} \sin\left(2\pi\left(u(t,v_j)-u(t,v_i)\right)\right), \quad
		v_i\in  V_n,
	\end{equation}
	where $u(t,v_i)$ represents the state of the oscillator located at $v_i$ at time $t$. The scaling factor $\left(\frac{5}{3}\right)^n$
	is used so that the model has a nontrivial continuum limit as $n\to\infty$.

	It is instructive to rewrite the KM on $\Gamma_n$ as the gradient system
	\begin{equation}\label{KM-grad}
		\dot u=-2\pi \nabla  \cJ_n(u), \quad u= (u(t,v_1),u(t,v_2),\dots, u(t,v_{k_n})),
	\end{equation}
	where
	\begin{equation}\label{KM-energy}
		\cJ_n(u) \doteq \left(\frac{5}{3}\right)^n\frac{1}{4\pi^2}
		\sum_{\{i,j\}\in E_n} \left( 1-\cos\left(2\pi(u(t,v_j)-u(t,v_i))\right)\right),
	\end{equation}
	where the summation is taken over all undirected edges of $\Gamma_n$.
	We use the convention that each (undirected) edge $\{i,j\}\in E_n$ appears only once in the sum to avoid double-counting. The factor $\frac{1}{4\pi^2}$ is taken for convenience of future analysis. 
	
	Since \eqref{KM-grad} is a gradient system, its attractor coincides with the set of minima of the energy
	function $\cJ_n$. Due to translation invariance of \eqref{KM-intro} the points of minimum of $\cJ_n$ are
	not isolated. To eliminate this translation invariance, we fix the value at $v_1$ in the steady-state solutions of \eqref{KM-intro}
	to $0$. This leads to the following minimization problem:
	\begin{equation}\label{min-intro}
		\cJ_n \;\rightarrow\; \min_{H_n},\quad H_n\doteq \{ u\in L(V_n, \T):\; u(v_1)=0\}.
	\end{equation}
	{That is we seek the function $u\in H_n$ that minimizes
		$\cJ_n$. Here,
		$\T\doteq\R/\Z$ is the unit circle and} $L(X,Y)$ denotes the space of functions with domain $X$ and co-domain $Y$.
	
	The main result of this paper reveals the link between the structure of the attractor of \eqref{KM-intro} and the
	topology of the SG. We show that there is one-to-one correspondence between the homotopy classes
	of continuous functions from the SG to $\T$ and stable equilibria of \eqref{KM-intro}.

	Before formulating the
	main result, we review the homotopy classes of continuous maps on SG following \cite{MM2024}.
	To this end, let
	$$
	\mathcal{P}_n=\left\{ \partial T_w,\; w\in S^n\right\} \quad \mbox{and}\quad
	\mathcal{P}= \bigcup_{n=0}^\infty   \mathcal{P}_n.
	$$
	
	{As above, order the loops $\partial T_w$ by their itineraries in lexicographical order and denote them by $\gamma_i$, e.g., $\gamma_0=\partial T, \gamma_1=\partial T_1$, $\gamma_2 = \partial T_2$, etc.}
	For each triangular loop $\gamma\in\mathcal{P}$, we choose a reference point $O_\gamma\in\gamma$.
	For concreteness, let $O_\gamma$ be the leftmost vertex of $\gamma$.    
	Choose a \textit{uniform} parametrization
	$c_\gamma:\T\to\gamma$, which starts at $O_\gamma$, $c_\gamma(0)=O_\gamma$, and traces $\gamma$ in clockwise
	direction with constant speed.

	For a given $f\in C(K,\T)$,  $f_\gamma\doteq f\circ c_\gamma$ is a continuous function
	from $\T$ to itself\footnote{Throughout this paper, $c_\gamma$ is a uniform parametrization, i.e.,
		$\dot c_\gamma=1.$}. There is a unique continuous function $\hat{f}_\gamma:\R\to\R$
	such that
	\begin{align} \label{lift-1}
		\hat{f}_\gamma(0)&= f_\gamma(0),\\
		\label{lift-2}
		\hat{f}_\gamma(x)&= f_\gamma(x\mod 1) + k(x), \quad k(x)\in \Z.
	\end{align}
	The second condition \eqref{lift-2} can be written as 
	\begin{equation}\label{lift-3}
		\pi\circ \hat f_\gamma=  f_\gamma \circ \pi,
	\end{equation}
	where $\pi: \,\R\ni x\mapsto x\mod 1.$ 
	
	$\hat{f}_\gamma$ is called the lift of $f_\gamma$. 
	The degree of $f_\gamma$ is expressed in terms of the lift of 
	$f_\gamma$:
	$$
	\omega(f_\gamma)\doteq \hat{f}_\gamma(1)-\hat{f}_\gamma(0).
	$$
	
	\begin{definition}\label{def.w-vector}
		The degree of $f\in C(K,\T)$ is defined by
		\begin{equation}\label{w-vector}
			\bar\omega(f)=\left(\omega(f_{\gamma_0}), \omega(f_{\gamma_1}), 
			\omega(f_{\gamma_2}),\dots \right).
		\end{equation}
	\end{definition}
	
	\begin{remark}
		Since $f$ is uniformly continuous on $K$, the number of nonzero entries in \eqref{w-vector}
		is finite. Unless specified otherwise, we use the topology on $K$ induced by the Euclidean metric.
	\end{remark}
	
	\begin{definition}\label{def.homotopy}
		Two maps $f,g \in C(K,\T)$ are called homotopic, denoted $f\sim g$, if there exists a continuous mapping
		$F: [0,1]\times K\to \T$ such that
		\begin{equation}\label{F-hom}
			F(0,\cdot)=f\qquad\mbox{and}\quad F(1,\cdot)=g.
		\end{equation}
	\end{definition}
	
	\begin{theorem}\label{thm.homotopy}(Theorem 2.4 in \cite{MM2024}) Let $f,g \in C(K,\T)$. Then
		$f\sim g$ if and only if $\bar\omega(f)=\bar\omega(g)$.
	\end{theorem}
	
	We will need the following result on harmonic maps from SG to the circle. It follows from \cite{Strich02}.
	For a given $\omega^\ast\in \Z^\ast \doteq \bigcup_{l=1}^\infty \Z^l$ there exists a unique solution $u^\ast \in C(K,\T)$
	of the following boundary value problem:
	\begin{align}\label{HM-Lap}
		\Delta u^\ast & =0,\\
		\label{HM-dir}
		u^\ast(v_1)&= 0,\\
		\label{HM-neu}
		\partial_{\bf n} u^\ast (v_i) &=0,  i=1,2,3, \\
		\label{HM-deg}
		\bar \omega(u^\ast)&=\omega^\ast,
	\end{align}
	where $\partial_{\bf n}$ stands for the normal derivative
        \cite{Str06}.
       {For functions on SG, partial derivatives are not defined in general.
          Instead, solutions to the Laplace equation in \eqref{HM-Lap} are understood as
          minimizers of the associated Dirichlet form \cite{Kig01}. Similarly, the normal
          derivative on SG is introduced via a discrete analogue of the Gauss–Green
          formula (see \cite{Kig01, Str06} for details). Formal definitions of the Laplacian
          and normal derivatives on SG are provided below; see Definition \ref{def.R-harm}
          and Lemma \ref{lem.natural}.}
	
	We can now formulate the main result of this work. Below, we will generalize this result for the KM
	on p.c.f. fractals, but for now we restrict our attention to the KM on SG.
	\begin{theorem}\label{thm.main-sg}
		Let $\omega^\ast\in \Z^\ast$ and $u^\ast\in C(K,\T)$ be the solution of the boundary value
		problem \eqref{HM-Lap}-\eqref{HM-deg}. Then, for each $\varepsilon>0$, there exists $N\in\N$ such
		that for all $n\ge N$, the KM on $\Gamma_n$
		has a stable steady state solution $u^{n}\in L(V_n,\T)$, which
		can be extended to a continuous function on $K$,
		$\tilde{u}^{n},$ in such a way that
		$\omega(\tilde u^{n})=\omega^\ast$ and
		\begin{equation}\label{harm-spline-estimate}
			\max_{x\in K} | \tilde{u}^{n}(x)-u^*(x)|<\epsilon.
		\end{equation}
	\end{theorem}

	{\begin{remark}
			In fact, the uniform norm in \eqref{harm-spline-estimate}
		can be replaced by the norm in the H\"older space $C^\beta(K)$ with
                $\beta=\frac{\log(5/3)}{\log 2}$ (see Corollary \ref{cor.upgrade}). 
        \end{remark}}
	
	Figure \ref{fig:snapshots} shows two examples of stable equilibria of KM on SG with non-trivial degrees.
	The KM equilibria in Theorem \ref{thm.main-sg} are
	analogous to twisted states of the KM on nearest-neighbor graphs (cf.~\cite{WilStr06};
	see also Section \ref{sec.near} for more details).
	Distinct twisted states are homotopic if and only if they share the same winding number
	$q$ so classifying twisted states is straightforward.
	On the other hand, on the fractal $K$ the homotopy of a map is determined by its
	degree \eqref{w-vector}, which can have arbitrary (finite) length. As a result, there is a rich
	diversity of KM equilibria of increasing topological complexity on the fractal $K$. 
	\begin{figure}[h]
		\centering
		{\bf a)} \includegraphics[width=.45\textwidth]{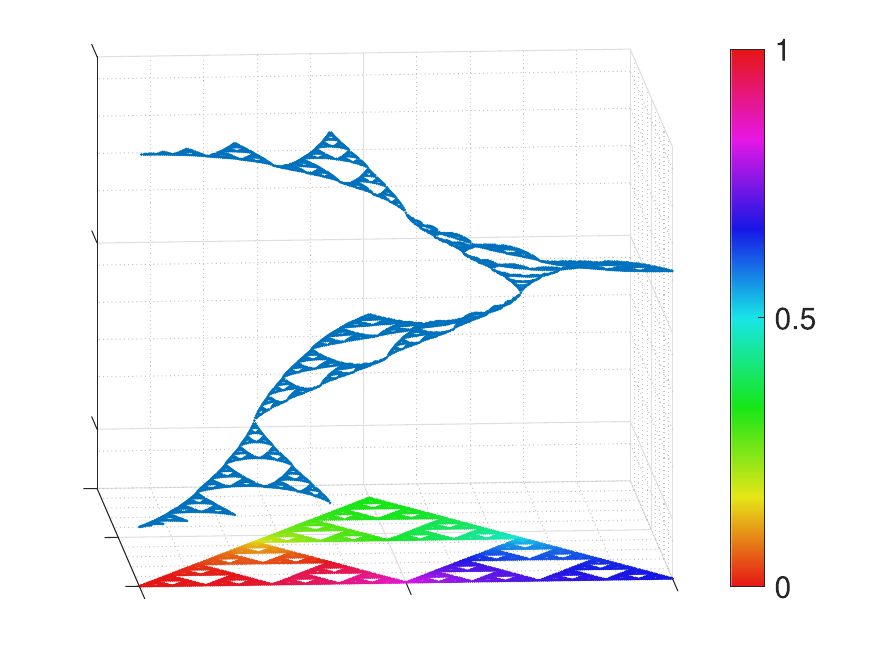}
		{\bf b)} \includegraphics[width=.45\textwidth]{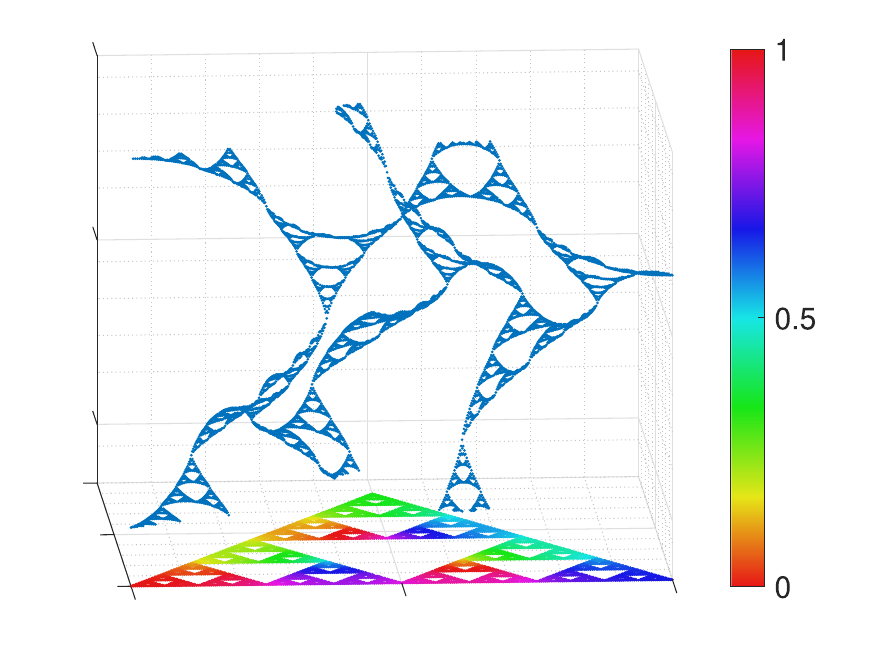}
		\caption{Twisted states on SG are equilibria of the KM on
			SG. The degrees of the equilibria shown above:
			{\bf a)} $\bar \omega(f)=(1,0,0,\dots)$, and {\bf b)} $\bar \omega(f)=(1,1,1,1,0,0,\dots)$.
			These plots were obtained by numerically solving the initial value problem for the KM \eqref{KM-intro} on
			on $\Gamma_8$ with initial conditions taken as the
			corresponding harmonic map $u^*$ satisfying \eqref{HM-Lap}-\eqref{HM-deg} and
			restricted to $V_8$.}
		\label{fig:snapshots}
	\end{figure}

	The remainder of this paper will be mainly focused on proving Theorem
	\ref{thm.main-sg}. We will begin by considering a simpler variation
	of the problem with trivial degree (i.e., $\bar
	\omega(u^*)=(0,0,\dots)$) and Dirichlet boundary conditions
	(cf. \eqref{HM-neu}). This problem is interesting in its own right,
	but also serves to build intuition and develop the necessary tools
	for the main result. In Section \ref{sec.harmSG} we review the
	necessary background on  harmonic functions on SG, and in Section
	\ref{sec.gamma} we combine $\Gamma$-convergence and a priori H\"older
	estimates to establish the convergence of KM equilibria to these
	(real-valued)  harmonic functions (Theorem \ref{thm:conv_minSG}). In
	Section \ref{sec.near} we translate these techniques to the setting
	of the KM on nearest-neighbor graphs. In particular, we show how
	covering spaces are used to adapt the results to $\T$-valued harmonic
	maps. We will then be prepared to prove Theorem \ref{thm.main-sg}. In
	Section \ref{sec.hmaps} we first review the necessary theory of
	$\T$-valued harmonic maps on SG, and in Section \ref{sec.proof} we
	complete the proof of
	Theorem~\ref{thm.main-sg}. Finally, in Section \ref{sec.pcf} we describe how these results can be applied in a more general setting to study the KM on post-critically finite fractals.

	\section{Harmonic functions on SG}\label{sec.harmSG}
	\setcounter{equation}{0}
	
	Harmonic functions can be defined as minimizers of the Dirichlet energy (cf.~\cite{Struwe-PDE}).
	On a Euclidean domain, the Dirichlet energy is given by
	\begin{equation}\label{Dir-Euc}
		\cI(u)=\frac{1}{2}\int_\Omega \left| \nabla u\right|^2 dx,\qquad \Omega\subset\R^d.
	\end{equation}
	On a fractal set, partial derivatives are not defined in general. Instead, the Dirichlet energy is defined
	as a limit of the functionals $\cE_n$, each associated with a sequence of graphs approximating a given fractal.
	(see Figure~\ref{fig:sg_graphs}). For the reader's convenience, we outline the limiting procedure
	below and refer the interested reader to \cite{Kig01} for details. 

	First, we construct the sequence of graphs $(\Gamma_m)$
	approximating SG (Figure~\ref{fig:sg_graphs}). 
	On $\Gamma_m$ we define the discrete Laplacian
	\begin{equation}\label{n-Laplace}
		\Delta_m f(v_i)=\left(\frac{5}{3}\right)^m\sum_{j:j\sim_m i} \left[f(v_j)-f(v_i)\right], \quad f\in L(V_m,\R).
	\end{equation}
	
	A function $f_m^* \in L(V_m,\R)$ is called $\Gamma_m$-harmonic if it satisfies the discrete
	Laplace equation at every interior node of $\Gamma_m$:
	\begin{equation}\label{d-harm}
		\Delta_mf_m^*(v)=0\quad \forall v\in V_m\setminus V_0.
	\end{equation}
	Equivalently, harmonic functions are critical points of the associated Dirichlet forms
	\begin{equation}
		\cE_m(f) \doteq \left(\frac{5}{3}\right)^m\sum_{\{i,j\}\in
			E_m} \frac{\left(f(v_j)-f(v_i)\right)^2}{2},
		\quad f\in L(V_m,\R).\end{equation}
	The sequence of Dirichlet forms $(\mathcal{E}_m)$ has the following 
	properties.
	\begin{enumerate}
		\item A $\Gamma_m$-harmonic $f_m^\ast\in L(V_m,\R)$ minimizes $\mathcal{E}_m$ 
		over all functions subject
		to the same boundary conditions
		\begin{equation}\label{variational-harmonic}
			\mathcal{E}_m(f_m^\ast)=\min\left\{ \mathcal{E}_m(f):\; f\in L(V_m,\R), \; f|_{V_0}=f_m^\ast|_{V_0}\right\}.
		\end{equation}
		\item The minimum of the energy over all extensions of 
		$f\in L(V_{m-1},\R)$ to $V_m$ is equal to $\mathcal{E}_{m-1}(f)$:
		\begin{equation}\label{variational-extension}
			\min\left\{ \mathcal{E}_m(\tilde f):\; \tilde f\in L(V_m,\R), 
			\tilde f|_{V_{m-1}}=f\in L(V_{m-1},\R)\right\} = \mathcal{E}_{m-1}(f).
		\end{equation} 
	\end{enumerate}
	The first property follows from the
	Euler-Lagrange equation for $\cE_m$. This property does not depend on the scaling
	coefficient $(5/3)^m.$ The second property, on the other hand, holds 
	due to the choice of the scaling constant.   The sequence of 
	$(\mathcal{E}_m)$ equips $K$ with a \textit{harmonic structure} 
        (cf.~\cite{Kig01}, {see also \cite{Str06}}).

   {
By standard results from the Electrical Networks theory (cf.~\cite{LyonsPeres}), there is a unique $\Gamma_m$-harmonic function $f_m^\ast \in L(V_m,\mathbb{R})$ satisfying given boundary conditions on $V_0$.

By \eqref{variational-extension}, the family of $\Gamma$-harmonic functions $\left(f_m^\ast,\;{m \in \mathbb{N}}\right)$ is compatible across scales:
\begin{equation}\label{marting}
f^\ast_{m+1}\big|_{V_m} = f^\ast_m.
\end{equation}
This compatibility is the basis for harmonic extension. 
}
    
	Let $f$ be a continuous function on $K$, $f\in C(K,\R)$. Define 
	$
	P_m: C(K,\R)\to L(V_m,\R)
	$
	by restricting $f$ to $V_m$:\;
	$
	P_m f\doteq f|_{V_m}.
	$
	We next extend the definition of $\cE_m$ to functions in $C(K,\R)$ (which we also denote by $\cE_m$ by abuse of notation):
	$$
	\cE_m(f):= \cE_m\circ P_m(f).
	$$
	By \eqref{variational-extension}, $(\cE_m(f))$ is a non-decreasing sequence. Thus,
	\begin{equation}\label{monotone-conv}
		\cE(f)\doteq \lim_{m\to\infty} \cE_m(f)
	\end{equation}
	defines a functional on $C(K,\R)$. 
	
	The domain of the Laplacian on $K$ is defined by
	$$
	\operatorname{dom}(\cE)=\left\{ f\in C(K,\R):\; \cE(f)<\infty \right\}.
	$$
	
	\begin{definition}\label{def.R-harm}
		A function $f\in \operatorname{dom}(\cE)$ is called harmonic,  if it 
		minimizes $\cE(f)$ over all continuous functions on SG subject to given boundary conditions 
		on $V_0$.
	\end{definition}

	The combination of \eqref{variational-harmonic} and \eqref{variational-extension}, implies that 
	the restriction of harmonic $f\in \operatorname{dom}(\cE)$ to $\Gamma_m$,
	$f\left|_{\Gamma_m}\right.$, is
	$\Gamma_m$-harmonic on $\Gamma_m$ for every $m\in\N$.

	The second property of the energy form (cf.~\eqref{variational-extension})  yields 
	a recursive algorithm for computing the values of a harmonic function on 
	$V_\ast=\bigcup_{m=0}^\infty V_m$, a dense subset
	of $K$. For $m=0$, the values on $V_0$ are prescribed:
	$$
	f|_{V_0}=\phi.
	$$
	Given $f|_{V_{m-1}},\; m\ge 1,$ the values on $V_m\setminus V_{m-1}$ are computed using the
	following $\frac{1}{5}-\frac{2}{5}$ rule, which we state for an arbitrary fixed $(m-1)$-cell
	$T_w,\; w\in S^{m-1}$: Suppose the values of $f$ at $a, b, c$, the nodes of $T_w$, are
	known. Then the values of $f$ at $x,y,z,$ the nodes at the next level of discretization  are 
	computed as follows 
	\begin{equation}\label{classical-extension}
		\begin{pmatrix} f(x) \\ f(y)\\ f(z)
		\end{pmatrix}
		=
		\begin{pmatrix}
			\frac{2}{5} &  \frac{2}{5} &  \frac{1}{5}\\
			\frac{1}{5} &  \frac{2}{5} &  \frac{2}{5}\\
			\frac{2}{5} &  \frac{1}{5} &  \frac{2}{5}
		\end{pmatrix}
		\begin{pmatrix} f(a) \\ f(b)\\ f(c)
		\end{pmatrix},
	\end{equation}
	see Fig.~\ref{fig:harm_ext}.
	\begin{figure}[h]
		\centering
		\includegraphics[width  =.3\textwidth]{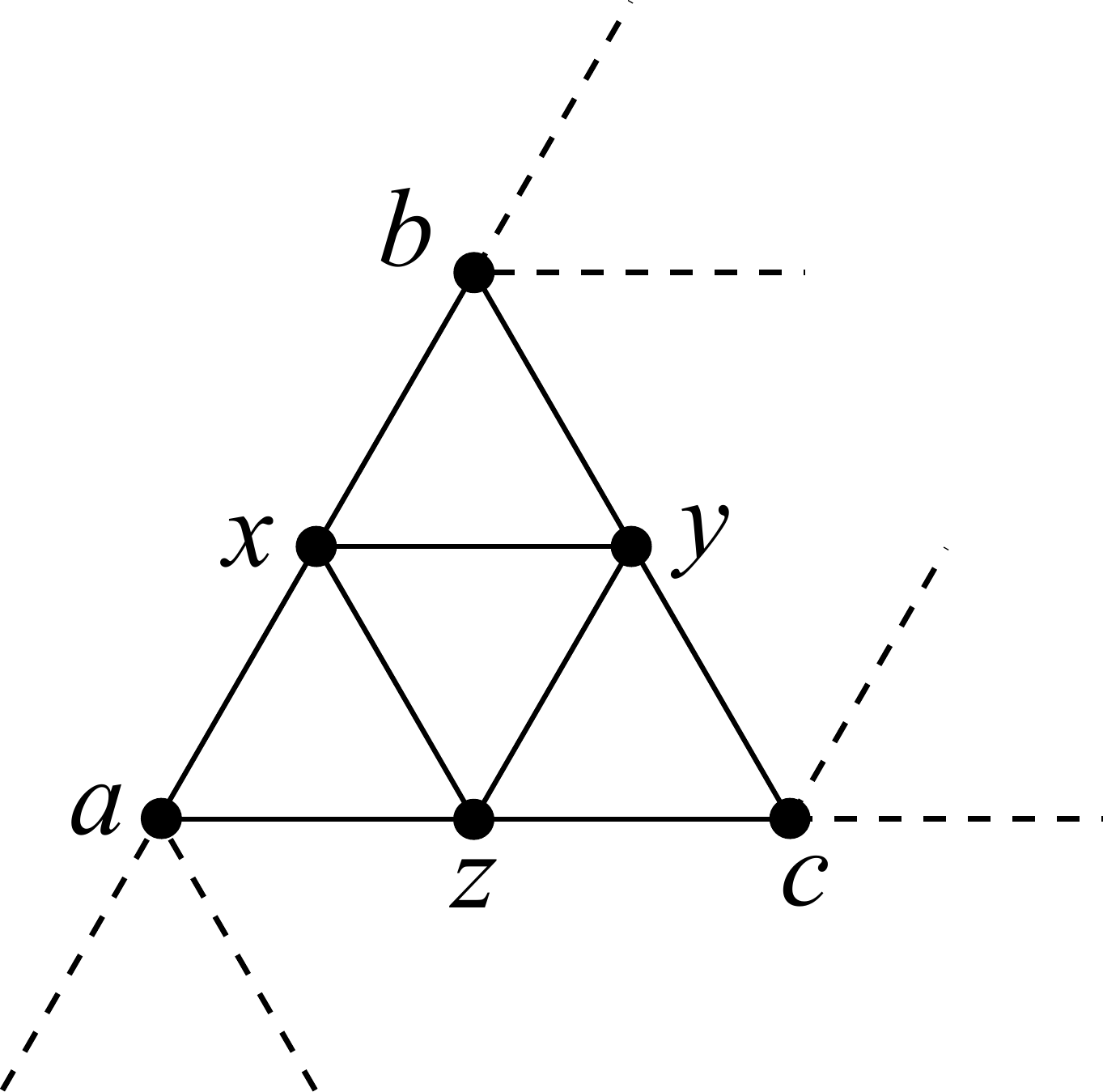}
		\caption{Harmonic extension algorithm for harmonic functions, see \eqref{classical-extension}.}
		\label{fig:harm_ext}
	\end{figure}
	
	This algorithm is called \textit{harmonic extension} \cite{Str06}.
	Using \eqref{variational-extension}, one can show that at each step of the harmonic extension:
	$$
	\cE_m(f|_{V_m})=\cE_0(f|_{V_0})=\min \left\{ \cE_m(f):\; f\in L(V_m,\R),\; f|_{V_0}=\phi\right\}.
	$$
	Note that each of  $f|_{V_m}$ is a $\Gamma_m$-harmonic function. Further, harmonic extension
	results in a uniformly continuous function on $V_*$. Thus, the values of $f$ on $K\setminus V_*$ can
	be obtained from those on $V_*$ by continuity.
	
	Finally, the following H\"older estimate will be used to provide compactness, which is needed for the
	convergence the equilibria of the KM.
	\begin{lemma}[cf. \cite{Str06,Kozlov1993}]\label{lem.SG_holder}
		If $f\in \operatorname{dom}(\cE)$ then
		\begin{equation*}
			\frac{|f(x)-f(y)|}{|x-y|^{\beta}}\leq C\sqrt{\mathcal{E}(f)},
		\end{equation*}
		where  $\beta = \frac{\log(5/3)}{2\log(2)}$ and $C$ is independent of $f$.
	\end{lemma}
	
	\section{$\Gamma$-convergence}\label{sec.gamma}
	\setcounter{equation}{0}
	
	$\Gamma$-convergence is our main tool for studying stable equilibria in the KM for large $n$.
	Below we recall the definition of $\Gamma$-convergence on metric spaces and refer the reader to \cite{Braides-handbook} for further details.
	
	
	\begin{definition} \label{def.Gamma} Let $(X,d)$ be a metric space and
		$F_n\colon X\to \overline{\R}:=\R\cup \{\pm \infty\}$ be a
		sequence of functionals and $F\colon X\to \overline{\R}.$
		$(F_n)$ is said to {\it $\Gamma$-converge} to $F$, denoted
		$\Gamma$-$\lim F_n = F$, if the following two conditions hold.
		\begin{enumerate}
			\item[C1.] For all $x_n\in X$ with $x_n\to x\in X$, $F(x)\leq \liminf_{n\to\infty} F_n(x_n)$.
			\item[C2.] For all $x\in X$, there is some $x_n\to x$ so that $F(x)\geq \limsup_{n\to\infty} F_n(x_n)$.
		\end{enumerate}
	\end{definition}
	
	\begin{remark} Sequence $(x_n)$ in C2 is called  a \textit{recovery sequence}.
		The combination of C1  and C2 implies that
		for every $x\in X$, there exists $x_n\to x$ such that
		\begin{equation*}
			F(x) = \lim_{n\to\infty} F_n(x_n).
		\end{equation*}
	\end{remark}
	
	{Fix Dirichlet boundary conditions $\phi$ on $V_0$. Let $X=\{u\in C(K,\R) : u|_{V_0} = \phi\}$} equipped with the supremum metric
	\begin{equation}\label{sup-metric}
		d(f,g)\doteq\sup_{x\in K} |f(x)-g(x)|.
	\end{equation}
	
	First, we note that $\cE_n$ $\Gamma$-converges to $\cE$ in this space ${(X,d)}$. This follows from the monotone
	pointwise convergence of $\cE_n\nearrow\cE$ (cf.~\eqref{monotone-conv}) and continuity of $\cE_n$
	(cf.~\cite[Remark~1.40(ii)]{Braides-beginners}). Importantly, $\cJ_n$
	also $\Gamma$-converges to $\cE$, as shown below.

	\begin{theorem} \label{thm:SGThm}
		{Fix Dirichlet boundary conditions $\phi$ on $V_0$. Let $X=\{u\in C(K,\R) : u|_{V_0}=\phi\}$ be equipped with the supremum norm \eqref{sup-metric}.} Let $(\cJ_n)$ be the Kuramoto energy functionals (cf.~\eqref{KM-energy}). Then
		\begin{equation}\label{Kuramoto-converge}
			\cE=\Glim{n\to\infty} \cJ_n.
		\end{equation}
		Moreover, suppose $(u^n)$ is a sequence of minimizers of  $(\cJ_n)$ converging to  $u^\ast\in X$, then
		$u^\ast$ is a minimizer of $\mathcal{E}$ and
		\begin{equation}\label{Kuramoto-minimizers}
			\lim_{n\to\infty} \cJ_n(u^n) = 	\cE(u^\ast).
		\end{equation}
	\end{theorem}
	
	\begin{proof}
		By the Taylor's expansion of $\cos$, 
		\begin{align}\nonumber
			\cJ_n(u) &= \left(\frac{5}{3}\right)^n \sum_{\{i,j\}\in E_n}  \frac{\left(u(v_j)-u(v_i)\right)^2}{2}
			\left[1+g(u(v_j)-u(v_i))\right]
		\end{align}
		where $g(x)=-\frac{2(2\pi x)^2}{4!}+\frac{2(2\pi x)^4}{6!}-\cdots$. Using the alternating
		series remainder bound,
		we have
		\begin{equation} \label{g-bound1}
			|g(x)|\leq \frac{\pi^2|x|^2}{3}.
		\end{equation}

		{To first prove the liminf inequality, suppose that $u^n, u\in C(K)$ with $u^n\to u$.
			Let $\varepsilon>0$. By construction of $\Gamma_n$, 
			$$
			|v_j-v_i | = 2^{-n}, \quad \mbox{ when } j\sim_n i.
			$$
			We note that $u$ is uniformly continuous, since $K$ is compact. Moreover, because
			$(u^n)$ converges uniformly to
			the uniformly continuous function $u$, the sequence $(u^n)$ is equicontinuous.
			Thus, there exists $N\in\N$ such that for all $n\geq N$,
			$|v_j-v_i|=2^{-n}$ implies
			$|u^n(v_j)-u^n(v_i)|<\varepsilon$ whenever $j\sim_n i$. Applying \eqref{g-bound1}, 
			\begin{equation*} 
				|g(u^n(v_j)-u^n(v_i)) |<C \varepsilon^2
			\end{equation*}
			for all $n\geq N$ and $j\sim_n i$.
		}
		
		{Thus, 
			\begin{equation*}
				\cJ_n(u^n) = \left(\frac{5}{3}\right)^n\sum_{\{i,j\}\in E_n}
				\frac{(u^n(v_j)-u^n(v_i))^2}{2}[1+g(u^n(v_j)-u^n(v_i))]\geq (1-C\varepsilon^2) \cE_n(u^n),\mbox{ for all } n\geq N.
			\end{equation*}
			Taking the liminf yields
			\begin{equation*}
				\liminf_n \cJ_n(u^n)\geq (1-C\varepsilon^2) \liminf_n\cE_n(u^n) \geq (1-C\varepsilon^2)\cE(u),
			\end{equation*}
			where we have used the $\Gamma$-convergence of $\cE_n$. Sending $\varepsilon\to 0$, we conclude $\liminf_n \cJ_n(u^n)\geq \cE_n(u)$.
		}
		This verifies the lower bound in the
		definition of $\Gamma$-convergence (cf.~C1. Definition~\ref{def.Gamma}).
		
		{To verify the upper bound, we need to show that for an arbitrary $u\in X$, there is a recovery sequence $u^n\to u$ such that $\limsup_n \cJ_n(u^n)\leq \cE(u)$ (cf. C2. Definition~\ref{def.Gamma}). To this end, let $u^n=u$ be the constant sequence.
			Arguing as above, we obtain the estimate
			\begin{equation*}
				\limsup_{n} \cJ_n(u^n)=	\limsup_n\cJ_n(u)\leq (1+C\varepsilon^2)\limsup_n \cE_n(u)\leq (1+C\varepsilon^2)\cE(u),
			\end{equation*}
			since $\cE_n$ is pointwise increasing. Again sending $\varepsilon\to 0$ we obtain $\limsup_n \cJ_n(u^n) \leq \cE(u)$, completing the proof of $\Gamma$-convergence of $\cJ_n$.}
		
		Finally, the fact that convergent sequences of minimizers converge to a minimizer of the
		$\Gamma$-limit is a standard result (see, e.g.,
		\cite{Braides-beginners}). Below we include a short proof of
		this statement for completeness.
		
		First, by $\Gamma$-convergence of $(\cJ_n)$, for $u \in X$ there is a recovery sequence
		$u^n\to u$ such that $\limsup_{n} \cJ_n(u^n)\leq \cE(u)$. Thus,
		\begin{equation*}
			\limsup_n \min_{w\in X} \cJ_n(w)\leq \limsup_n \cJ_n(u^n)\leq \cE(u).
		\end{equation*}
		Since $u\in X$ is arbitrary,
		\begin{equation}\label{eq:limsupEst1}
			\limsup_n \min_{w\in X} \cJ_n(w) \leq \inf_{u\in X} \mathcal{E}(u).
		\end{equation}
		
		Suppose now that $u^n$ is a sequence of minimizers of $\cJ_n$ with $u^{n}\to u$.
		Using $\Gamma$-convergence  of $(\cJ_n)$ again, we have
		\begin{equation*}
			\inf_{X} \mathcal{E} \leq \mathcal{E}(u) \leq \liminf_n\cJ_n(u^n) =
			\liminf_n \min_{X} \cJ_n\leq \limsup_n \min_{X} \cJ_n\leq \inf_{X}\cE,
		\end{equation*}
		where the final inequality follows from \eqref{eq:limsupEst1}. Thus,
		$$
		\min_{X} \cE=\cE(u) = \lim_n \cJ_n(u^n).
		$$
	\end{proof}
	
	The key implication of the $\Gamma$-convergence of $(\mathcal{E}_n)$ and $(\mathcal{J}_n)$ to $\mathcal{E}$ is that any limit point of a sequence of minimizers
	of $(\mathcal{E}_n)$ or $(\mathcal{J}_n)$ is a minimizer of $\mathcal{E}$, i.e., a harmonic function on
	the SG. For $(\mathcal{E}_n)$, one can easily find the minimizers by either solving
	a system of linear equations
	or by harmonic extension. It is straightforward to show that the
	sequence of these minimizers converges to a harmonic function on SG.
	The situation for the Kuramoto energy functionals $(\mathcal{J}_n)$ is more subtle: the minimization
	does not reduce to a linear system, and the harmonic extension method is not applicable. Nevertheless,
	using general properties of $\Gamma$-convergence, below we show the existence of a sequence
	of minimizers of $(\mathcal{J}_n)$ that converges to a harmonic function on SG.

	\begin{theorem}\label{thm:conv_minSG}
		Let $u^*$ be the harmonic function on the SG subject to the following boundary condition
		$u^*|_{V_0}=\phi\in L(V_0,\R)$.
		Then there exists a sequence of stable equilibria of the KM on $\Gamma_n$, $u^n$, with $u^n|_{V^0}=\phi$ so that
		$\max_{x\in V^n}|u^n(x)-u^*(x)|\to 0$ as $n\to\infty$.
	\end{theorem}
	
	\begin{proof}
		Fix $\varepsilon>0$. We first show that, for all sufficiently large $n$,
$\mathcal J_n$ has a local minimum in
\[
B_\varepsilon(u^*)
=
\left\{
f\in C(K):
f|_{V_0}=\phi,\ 
\|f-u^*\|_{C^0}<\varepsilon
\right\}.
\]

A continuous function on $K$ whose restriction to each $n$-cell of $K$ is
harmonic is called a harmonic spline. Consider the space of harmonic splines
\begin{equation}
X^n
:=
\left\{
f\in C(K,\mathbb R):
f\circ F_w\text{ is harmonic for all }w\in S^n
\right\}.\label{harm-splines}
\end{equation}
The space $X^n$ is finite-dimensional. Hence
$X^n\cap\overline{B_\varepsilon(u^*)}$ is compact as a closed and bounded
subset of the finite-dimensional space $X^n$.

For every $n\in\mathbb N$, $\mathcal J_n$ is continuous on $C(K)$. Thus it
attains its minimum on
$X^n\cap\overline{B_\varepsilon(u^*)}$. Denote a minimizer by $u^n$.
If the minimum is attained at multiple points, choose, whenever possible, a
minimizer lying in $B_\varepsilon(u^*)$.

It remains to show that $u^n\in B_\varepsilon(u^*)$ for all but finitely many
$n$. Suppose, to the contrary, that there are infinitely many $n$ for which
\[
\|u^n-u^*\|_{C^0}=\varepsilon.
\]
Passing to a subsequence and relabeling, we may assume that this equality
holds for every $n$ under consideration.

By uniform continuity of $u^*$, there exists $N\in\mathbb N$ such that, for
$n\geq N$,
\[
v_i\sim_n v_j
\quad\Longrightarrow\quad
|u^*(v_i)-u^*(v_j)|<\varepsilon.
\]
Consequently,
\begin{align*}
|u^n(v_i)-u^n(v_j)|
&\leq |u^n(v_i)-u^*(v_i)|
 +|u^*(v_i)-u^*(v_j)|
 +|u^*(v_j)-u^n(v_j)| \\
&<3\varepsilon,
\qquad n\geq N.
\end{align*}

Proceeding as in the proof of Theorem \ref{thm:SGThm}, we obtain
\[
\mathcal J_n(u^n)
\geq
(1-C\varepsilon^2)\mathcal E_n(u^n).
\]
Choose $\varepsilon$ sufficiently small that $C\varepsilon^2<1/2$.
Then
\[
\mathcal E_n(u^n)
\leq
\frac{1}{1-C\varepsilon^2}\mathcal J_n(u^n)
\leq
2\mathcal J_n(u^n).
\]
Since $u^n$ minimizes $\mathcal J_n$ on
$X^n\cap\overline{B_\varepsilon(u^*)}$ and $u^*$ belongs to this set,
\[
\mathcal J_n(u^n)\leq\mathcal J_n(u^*).
\]
Moreover, since $u^*$ is harmonic,
\[
\mathcal J_n(u^*)\leq\mathcal E_n(u^*)=\mathcal E(u^*).
\]
Because $u^n$ is a harmonic spline,
\[
\mathcal E_n(u^n)=\mathcal E(u^n).
\]
Therefore
\[
\mathcal E(u^n)\leq 2\mathcal E(u^*),
\]
so $\{\mathcal E(u^n)\}$ is bounded.

By Lemma~\ref{lem.SG_holder}, the sequence $\{u^n\}$ has uniformly bounded
$C^\beta$ seminorms. Since
\[
\|u^n\|_{C^0}
\leq
\|u^*\|_{C^0}+\varepsilon,
\]
the sequence is uniformly bounded in $C^\beta(K)$. By the compact embedding
$C^\beta(K)\hookrightarrow C^0(K)$, there is a subsequence, still denoted
by $u^{n_k}$, such that
\[
u^{n_k}\longrightarrow\bar u
\qquad\text{in }C^0(K).
\]

Since $u^*$ is the unique minimizer of $\mathcal E$ among functions satisfying
the prescribed boundary condition, we have
\[
\mathcal E(u^*)\leq\mathcal E(\bar u).
\]
On the other hand, by the $\Gamma$-liminf inequality,
\[
\mathcal E(\bar u)
\leq
\liminf_{k\to\infty}\mathcal J_{n_k}(u^{n_k})
\leq
\liminf_{k\to\infty}\mathcal J_{n_k}(u^*)
=
\mathcal E(u^*).
\]
Thus
\[
\mathcal E(\bar u)=\mathcal E(u^*).
\]
By uniqueness of the harmonic energy minimizer,
\[
\bar u=u^*.
\]
This contradicts
\[
\|u^{n_k}-u^*\|_{C^0}=\varepsilon.
\]
Hence $u^n\in B_\varepsilon(u^*)$ for all sufficiently large $n$, completing
the proof.
	\end{proof}
	
	\begin{remark}
		Both Theorem~\ref{thm:conv_minSG} and Theorem~\ref{thm.main-sg} claim
		the existence of families of equilibria in the discrete KM models converging to harmonic
		functions on~$K$.
		However, the settings of the two theorems differ. In the former case,
		we consider real-valued solutions subject to Dirichlet boundary conditions, whereas in
		the latter, we deal with $\T$-valued solutions for the model with free boundary
		conditions. Nonetheless, the proof of Theorem~\ref{thm:conv_minSG} provides a roadmap
		that will be useful in addressing the more challenging setting of Theorem~\ref{thm.main-sg}.
	\end{remark}

\begin{cor}\label{cor.upgrade}
The sequence of equilibria $u^n$ constructed in
Theorem~\ref{thm:conv_minSG} converges to $u^*$ in $C^\beta(K)$.
\end{cor}

\begin{proof}
From Theorem~\ref{thm:conv_minSG},
\[
u^n\longrightarrow u^*
\qquad\text{in }C^0(K).
\]
By Lemma~\ref{lem.SG_holder},
\[
\|u^n-u^*\|_\beta
\leq
\|u^n-u^*\|_{C^0}
+
C\sqrt{\mathcal E(u^n-u^*)}.
\]
It therefore suffices to prove that
\[
\mathcal E(u^n-u^*)\longrightarrow0.
\]

First, using the $\Gamma$-liminf inequality and the minimizing property of
$u^n$, we have
\begin{align*}
\mathcal E(u^*)
&\leq
\liminf_{n\to\infty}\mathcal J_n(u^n)\\
&\leq
\limsup_{n\to\infty}\mathcal J_n(u^n)\\
&\leq
\limsup_{n\to\infty}\mathcal J_n(u^*)\\
&\leq
\limsup_{n\to\infty}\mathcal E_n(u^*)\\
&=
\mathcal E(u^*).
\end{align*}
Here, the third inequality follows from the minimizing property of $u^n$,
and the fourth follows from the harmonicity of $u^*$ and
\[
\mathcal J_n(u^*)\leq\mathcal E_n(u^*)=\mathcal E(u^*).
\]
Thus all inequalities are equalities, and
\[
\mathcal J_n(u^n)\longrightarrow\mathcal E(u^*).
\]

We next show that
\[
\mathcal E_n(u^n)-\mathcal J_n(u^n)\longrightarrow0.
\]
Fix $\varepsilon>0$ sufficiently small that $C\varepsilon^2<1/2$. Since
$u^n\to u^*$ in $C^0(K)$, there exists $N=N(\varepsilon)$ such that
\[
\|u^n-u^*\|_{C^0}<\varepsilon,
\qquad n\geq N.
\]
The estimate used in the proof of Theorem~\ref{thm:conv_minSG} then gives
\[
(1-C\varepsilon^2)\mathcal E_n(u^n)
\leq
\mathcal J_n(u^n),
\qquad n\geq N.
\]
The energies $\mathcal E_n(u^n)$ are bounded independently of $n$; indeed,
the preceding estimate and
$\mathcal J_n(u^n)\leq\mathcal J_n(u^*)\leq\mathcal E(u^*)$
give
\[
\mathcal E_n(u^n)\leq 2\mathcal E(u^*).
\]
Therefore,
\[
0\leq
\mathcal E_n(u^n)-\mathcal J_n(u^n)
\leq
C\varepsilon^2\mathcal E_n(u^n)
\leq
C_2\varepsilon^2,
\qquad n\geq N,
\]
where $C_2$ is independent of $\varepsilon$ and $n$.
Consequently,
\[
\limsup_{n\to\infty}
\left|
\mathcal E_n(u^n)-\mathcal J_n(u^n)
\right|
\leq
C_2\varepsilon^2.
\]
Since $\varepsilon>0$ is arbitrary, we obtain
\[
\mathcal E_n(u^n)-\mathcal J_n(u^n)\longrightarrow0.
\]

Since $u^n$ is a harmonic spline,
\[
\mathcal E_n(u^n)=\mathcal E(u^n).
\]
Combining this identity with
$\mathcal J_n(u^n)\to\mathcal E(u^*)$ yields
\begin{equation}\label{eq:energy_conv0}
\mathcal E(u^n)\longrightarrow\mathcal E(u^*).
\end{equation}

Finally, write $\mathcal E$ as the quadratic form associated with its
symmetric bilinear form. Then
\begin{equation}\label{eq:energy_conv1}
\mathcal E(u^n-u^*)
=
\mathcal E(u^n)-\mathcal E(u^*)
-2\mathcal E(u^*,u^n-u^*).
\end{equation}
Fix $m\in\mathbb N$. Since $u^*$ is harmonic, the standard energy identity
(see \cite[Lemma~1.3.1]{Strich2006}) gives
\begin{align}
\mathcal E(u^*,u^n-u^*)
&=
\mathcal E_m(u^*,u^n-u^*) \nonumber\\
&=
\left(\frac{5}{3}\right)^m
\sum_{x\sim_m y}
\bigl(u^*(y)-u^*(x)\bigr)
\bigl((u^n-u^*)(y)-(u^n-u^*)(x)\bigr).
\label{eq:energy_conv2}
\end{align}
For fixed $m$, the sum contains finitely many terms, and
$u^n-u^*\to0$ uniformly. Hence
\[
\mathcal E(u^*,u^n-u^*)\longrightarrow 0.
\]
Combining this with \eqref{eq:energy_conv0} and \eqref{eq:energy_conv1}, we
obtain
\[
\mathcal E(u^n-u^*)\longrightarrow 0.
\]
The desired convergence in $C^\beta(K)$ follows.
\end{proof}

	\section{The KM on nearest-neighbor graphs}\label{sec.near}
	\setcounter{equation}{0}
	
	In the previous section, we proved $\Gamma$-convergence of the Kuramoto energy to the Dirichlet
	energy as $n \to \infty$.
	Until now, we have treated solutions of the KM as real-valued
	functions. It is time to recognize the importance of viewing them as maps from their respective
	domains to the circle. To extend $\Gamma$-convergence techniques to $\T$-valued maps, we
	incorporate several geometric constructions developed in~\cite{MM2024}. To make the
	transition to the $\T$-valued setting more intuitive, we first analyze the KM on the unit circle.
	The fact that a circle, an interval, and a cube are self-similar sets was used effectively
	in~\cite{Str06} to elucidate analysis on fractals.
	Similarly, we begin the analysis of the KM on SG by
	first considering it on nearest-neighbor graphs.
	Besides its pedagogical value, the KM
	on nearest-neighbor graphs is of independent interest, as it features prominently in many applications
	(see, e.g.,~\cite{WilStr06, BMS2025, MedZhu12}).
	
	In this section, consider the KM on the nearest-neighbor graph
	\begin{equation}\label{KM-nn}
		\dot u^n(t,v_i)=2^n\sum_{j:j\sim_n i} \sin\left(
		2\pi(u^n(t,x_j)-u^n(t,x_i))\right),\quad
		x_i=\frac{i}{2^n},\;   i\in\left\{0, 1, \dots, 2^n-1 \right\},
	\end{equation}
	where
	\begin{equation}\label{nn-adj}
		j\sim_n i \Leftrightarrow j=i\pm 1 \mod 2^n.\end{equation}
	
	Since we are interested in the continuum limit of \eqref{KM-nn} as $n\to\infty$, it is convenient
	to interpret $u^n(t,v_i)$ as the values of $u^n(t,x)$ at $x=i2^{-n}$ where $i$ runs from $0$ to $2^n-1$.
	Specifically,
	we represent the spatial domain of $u^n(t,x)$ as a discretization of the unit interval with the two end points
	identified
	$I=[0,1]\setminus \{0\simeq 1\}$ by the set of points
	\begin{equation}\label{vertex-nn}
		V_n=\{0, 2^{-n}, 2\cdot 2^{-n},\dots,
		(2^{n}-1)2^{-n}\}=:\{x_0, x_1,\dots, x_{2^n-1}\}.
	\end{equation}
	Accordingly, we view \eqref{KM-nn} as a dynamical system on the nearest neighbor graph $\Gamma_n$
	with the vertex set $V_n$ \eqref{vertex-nn} and the adjacency relation \eqref{nn-adj}.
	
	We will argue that the natural candidate for the continuum limit of \eqref{KM-nn}
	is the heat equation
	\begin{equation}\label{heat-c}
		\partial_t u=\Delta u, 
	\end{equation}
	on $(0,1)$ with periodic boundary condition
	$
	u(0)=u(1)
	$
	\footnote{In \cite{Groi2025}, the heat equation was
		identified as a continuum limit of the 
		KM on geometric random graphs.}.
	Here, we view $u(t,\cdot)$ as a map from $\T$ to $\T$. Thus, the steady states of \eqref{heat-c}
	are harmonic maps from $\T$ to $\T$:
	\begin{equation}\label{1d-hm}
		\Delta u =0, \quad u(0)=u(1)=0.
	\end{equation}
	Here, we set $u(0)=0$ to eliminate translation invariance.
	
	If $u$ were real-valued then $u\equiv 0$ would be a unique solution of \eqref{1d-hm}. In the class
	of $\T$-valued functions, \eqref{1d-hm} has infinitely many solutions:
	\begin{equation}\label{twist}
		u(x)= qx\mod 1,\quad q\in\Z,
	\end{equation}
	i.e.,   one solution per homotopy class, determined by the degree  $q$. Such solutions are called
	$q$-twisted states after \cite{WilStr06}.
	
	Going back to the KM, it is easy to find discrete counterparts of \eqref{twist}:
	\begin{equation}\label{disc-twist}
		u(x_i)=qi2^{-n}\mod 1, \quad i\in [2^n].
	\end{equation}
	In the case of the KM on the circle, \eqref{disc-twist} clearly converges to \eqref{twist} as $n \to \infty$.
	However, for other self-similar sets, e.g., for SG, the relation between the solutions of the KM and those
	of its continuum limit is not as transparent. 
	In preparation for the analysis of the KM on SG and other
	fractals, below we extend the $\Gamma$-convergence techniques from the previous section
	to cover $\T$-valued solutions.
	
	\begin{remark}
		Below, we focus on nearest-neighbor ring graphs to make the connection with the KM on
		fractals, which is our ultimate goal. However, the analysis in the remainder of this section can
		be naturally extended to cover $k$-nearest-neighbor coupling. If $k\gtrsim O(1)$
		the results can be further extended to random networks
		(e.g., small-world graphs) using
		approximation results for dynamical models on random networks \cite{Med19}.
		In \cite{Groi2025, CirGro2024}, twisted states on geometric random graphs are
		analyzed using a different, albeit related, approach.
	\end{remark}
	
	\subsection{Harmonic maps on $\T$}\label{sec:mapsonT}
	Before we begin, we recall a few basic facts about continuous maps from $\T$ to $\T$ (see also Section \ref{sec.intro}).
	For a given $f\in C(\T,\T)$, there is a unique $\hat f\in C(\R,\R)$ such that
	$$
	\pi\circ \hat f(x)= f\circ \pi (x)\quad \forall x\in\R,
	$$
	where $\pi(x)\doteq x\mod 1$. $\hat f$ is called the \textit{lift}  of $f$.
	
	The integer $\hat f(1)-\hat f(0)$ is called the \textit{degree of} $f$. The degree counts the number of
	rotations around the circle made by  the trajectory $(f(t), \;0\le t\le 1)$.
	By the Hopf degree theorem, the degree determines the homotopy class
	of continuous maps from $\T$ to $\T$.
	
	We now explain our strategy for extending the $\Gamma$-convergence techniques for the model at hand.
	\begin{description}
		\item[I.] Let $q\in\Z$ be fixed. For a homotopy class defined by $q$, we  construct the
		covering of the circle. Define
		\begin{align*}
			\T_\times & \doteq \T\times \Z,\\
			\T^k&\doteq \T\times \{k\}\subset \T_\times, \; k\in \Z\qquad\mbox{(sheets)}
		\end{align*}
		Cut $\T^k$ at $0\times \{k\}$ producing two copies $0^k$ and $1^k$. Identify $1^k\simeq 0^{k+q}$. Then $\tilde{\T}\doteq \T_\times/\simeq$ is the covering space, and $\T^0$ is called the \textit{fundamental domain}, see Figure \ref{fig:cs_circle}. Discretize $\T^0$ by $V_n^0 = \{0,2^{-n},\dots, (2^n-1)2^{-n},1\}$, where we omit superscripts on vertex points for notational ease. 
		\begin{figure}[h]
			\centering
			\includegraphics[width = .8\textwidth]{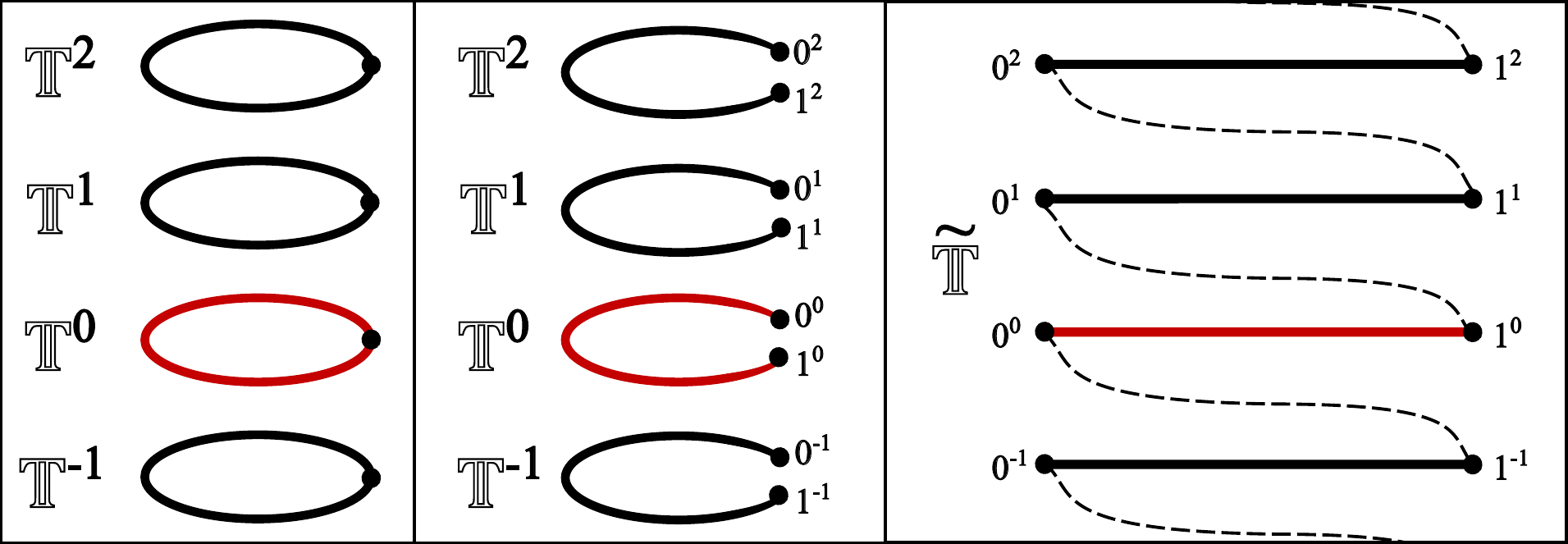}
			\caption{Construction of the covering space $\tilde{\T}$ for $q=1$. The fundamental domain $\T^0$ is highlighted in red.}
			\label{fig:cs_circle}
		\end{figure}
		\item[II.] On the fundamental domain, reformulate the boundary value problem for the discrete Laplacian
		for the lift of $u$:
		\begin{equation}\label{bvp-fundamental}
			\Delta_n \hat u^{n}=0, \quad \hat u^{n}(0)=0, \; \hat u^{n}(1)=q.
		\end{equation}    
		In dimension~$1$, this problem can be solved immediately:
		\begin{equation}\label{discrete-twist}
			\hat u^{n}(x_i)=qi 2^{-n},\quad i\in \{0,1,\dots, 2^n-1\}.
		\end{equation}
		\item[III.] We have computed the value of solutions on a dense subset of the
		fundamental domain $V^\ast=\bigcup_{n\in \N} V^0_n$. Noting that the discrete twisted states
		\eqref{discrete-twist} are Lipschitz continuous on $V^\ast$, we extend by them by continuity
		to the rest of the fundamental domain, to obtain the harmonic
		function on the fundamental domain, denoted
		$
		\hat u^{\ast}.
		$
		Extend these solutions to other sheets of the covering space:
		$$
		{\bf \hat u^{\ast}}(x,k)=\hat u^{\ast}(x)+k, \quad k\in\Z.
		$$
		Then ${\bf \hat u^{\ast}}$ is a harmonic function on the covering space.
		\item[IV.] Finally, project the range of the harmonic function ${\bf \hat u^{\ast}}$ back to $\T$-values:
		\begin{equation*}
			u^* := {\bf \hat u^{\ast}}|_{\T^0}\mod 1.
		\end{equation*}
		Then $u^*$ is the desired harmonic map on $\T$. 
	\end{description}

	\subsection{$\Gamma$-convergence on the covering space}
	
	In Section~\ref{sec.gamma}, we demonstrated convergence of the energy functional in the KM
	to the Dirichlet energy, from which we derived convergence of stable equilibria of the KM to
	harmonic functions. To derive these statements to the KM on the nearest-neighbor graph, we first need
	to translate the analysis to the real-valued setting. To this end, we use the procedure described in the previous subsection.
	Specifically, we first lift the $\T$-valued solutions on $\T$ and their discretizations to $\R$-valued functions
	on the covering spaces. Then we replicate the arguments from
	Section~\ref{sec.gamma} to obtain desired convergence of the energy functionals and the equilibria on the fundamental domain, $\T^0$. 
	
	Discrete harmonic functions $ \hat u^{n}$ minimize
	\begin{equation*}
		\cE_n({u}) := 2^n \sum_{(i,j)\in E_n} \frac{({u}(v_j)-{u}(v_i))^2}{2},\quad {u} = ({u}(x_0),{u}(x_1),\dots, {u}(x_{2^n}))
	\end{equation*}
	over $X_q^n=\{  f\in\R^{2^n+1}:\;  f(x_0)=0,\;  f(x_{2^n})=q\}$, the discrete approximation of the fundamental domain $\T^0$ with fixed jump condition.
	
	Let $X_q$ be the space
	\begin{equation*}
		X_q=\{ f\in C([0,1],\R) :  f(0)=0,\;  f(1)=q\}.
	\end{equation*}
	As in \eqref{monotone-conv}, the limit
	\begin{equation}\label{eq:e_nn}
		\cE( f) = \lim_{n\to\infty} \cE_n( f)
	\end{equation}
	is well-defined. The domain of $\cE$ is  {$\operatorname{dom}(\cE)=\{f\in H^1([0,1]): f(0)=0,\; f(1)=q\}$} \cite{Str06}.

	To apply the techniques from the proof of Theorem~\ref{thm:SGThm}, we need only
	establish H\"older continuity of functions in the domain of $\cE$, cf. Lemma \ref{lem.SG_holder}.
	By Morrey's inequality  \cite{brezis2011functional}, we have 
	\begin{equation*}
		\|{u}\|_{1/2}\leq C\|{u}\|_{H^1([0,1])},
	\end{equation*}
	where $\|\cdot\|_{1/2}$ is the $1/2$-H\"older norm, and $C$ is independent of $u$.
	Thus,
	$$
	\frac{|{u}(x)-{u}(y)|}{|x-y|^{1/2}}  \leq 	\|{u}\|_{1/2} \leq C_1 \|{u}\|_ {H^1([0,1])}
	\le C_2 \|{u}'\|_{L^2([0,1])} = C_2 \sqrt{\mathcal{E}({u})}.
	$$
	Here, we also used the Poincar\'e inequality to bound $\|{u}\|_ {H^1([0,1])}$.

	Now we are in a position to use the arguments from Section~\ref{sec.gamma}.
	Fix the space {$X_q$} with the supremum metric. As before, $\Gamma-\lim \cE_n=\cE$ due to
	the monotone pointwise convergence $\cE_n\nearrow \cE$ and continuity of the $\cE_n$.
	
	Next, we define the Kuramoto energy on the fundamental domain $\T^0$ 
	\begin{equation} 
		\cJ_n({u}):= 2^n \cdot\frac{1}{4\pi^2}\sum_{(i,j)\in E_n} (1-\cos(2\pi ({u}(v_j)-{u}(v_i)))).\label{eq:km_energy_nn}
	\end{equation} 
	Proceeding as in Theorem \ref{thm:SGThm}, we obtain $\Gamma$-convergence for (real-valued) energies on the fundamental domain.
	\begin{theorem}\label{thm.nn-graphs}
		{\emph{Fix $X_q$} equipped with the supremum norm.} Let $(\cJ_n)$ be the Kuramoto energy functionals \eqref{eq:km_energy_nn} and 
		let $\cE$ be given by \eqref{eq:e_nn}. Then
		\begin{equation}\label{Kuramoto-converge_nn}
			\cE=\Glim{n\to\infty} \cJ_n.
		\end{equation}
		Moreover, suppose $({u}^n)$ is a sequence of minimizers of  $(\cJ_n)$ converging to  ${ u}^\ast\in X_q$. Then $ u^\ast$ is a minimizer of $\mathcal{E}$ and
		\begin{equation}\label{Kuramoto-minimizers_nn}
			\cE({ u}^\ast) = \lim_n \cJ_n({ u}^n).
		\end{equation}
	\end{theorem}
	\begin{remark}
		Theorem \ref{thm.nn-graphs} establishes convergence of (real-valued) minimizers on
		the covering space $\T^0$. However, by periodicity of $\sin$ and $\cos$, it is clear that
		the minimizers ${ u}^n$ remain minimizers of the associated $\T$-valued energy when
		projected to harmonic maps ${ u}^n \mod 1$.
	\end{remark}
	
	The analogue of Theorem \ref{thm:conv_minSG} is trivial in this setting. Indeed, it is known
	that for all $n$ sufficiently large (depending on $q$), the discrete twisted states \eqref{disc-twist}
	are stable equilibria of the KM. Thus, the minimizers of $\cE_n$ and $\cJ_n$ coincide for sufficiently
	large $n$.
	
	\begin{remark}\label{ex.half-twist}
		It is interesting to relate Theorem~\ref{thm.nn-graphs} to the stability analysis of equilibria
		in the KM presented in \cite{BMS2025}. In addition to the twisted states, whose stability can be
		established using the discrete Fourier transform (cf.~\cite[Theorem 3.4]{MedTan2015a}),
		there exists another class of equilibria -- half-twisted states:
		\begin{equation*}
			{u}^{(r,n)}(v_i)=  \frac{r}{2^n-2} i,\qquad i\in \{1,\dots,2^n\},
		\end{equation*}
		where $r\in \Z+1/2$ is a half-integer satisfying $-\frac{2^n}{4}+\frac{1}{2}<r<\frac{2^n}{4}-\frac{1}{2}$ (see Fig.~\ref{fig:half_twist}).
		
		Denote the limit of ${u}^{(r,n)}$ by ${u}^r$. If the sequence $({u}^{(r,n)})$ were composed of minimizers of
		the Kuramoto energy $\cJ_n$, then by Theorem~\ref{thm.nn-graphs}, the limit $u^r$ would be a minimizer
		of $\cE$. However, since $u^r$ is discontinuous, it lies outside \textcolor{red}{$X_q$}.
		We therefore conclude that ${u}^{(r,n)}$ cannot be minimizers of $\cJ_n$ for sufficiently large $n$.
		This is consistent with the results of \cite{BMS2025}, which show that ${u}^{(r,n)}$ are saddles.
	\end{remark}

	\begin{figure}[h]
		\centering
		\includegraphics[width = .5\textwidth]{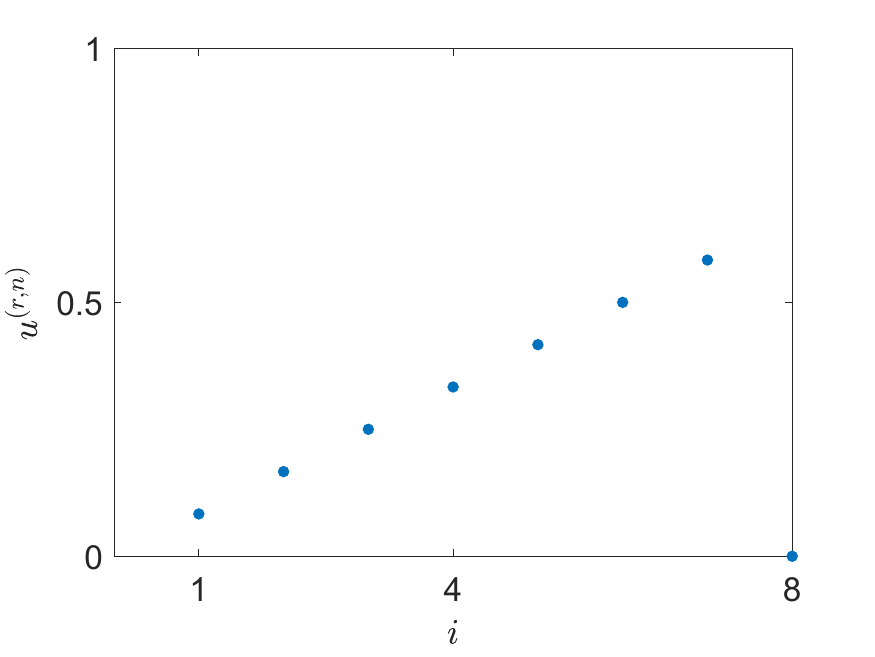}\quad \quad 
		\caption{A half-twisted state $u^{(r,n)}$ as described in Remark \ref{ex.half-twist}. Here $r=1/2$ and $n=3$. }
	\label{fig:half_twist}
\end{figure}

\section{Harmonic maps from SG to the circle}\label{sec.hmaps}
\setcounter{equation}{0}

Extending the ideas of the previous section, we aim to show that the stable equilibria of the
KM on graphs approximating SG converge to minimizers of the Dirichlet energy for $\T$-valued functions
on SG, i.e., harmonic maps. As before, our approach relies on $\Gamma$-convergence
techniques, now adapted to the setting of $\T$-valued functions on SG. To this end, we construct an
appropriate covering space for SG and lift $\T$-valued functions on SG to real-valued functions on this
covering space. We then apply $\Gamma$-convergence to show that the stable equilibria in the KM are
close to harmonic maps, which minimize the Dirichlet energy on the covering space -- more precisely,
on a fundamental domain. Finally, we project back to SG to establish convergence of the  equilibria
of the KM to harmonic maps.

In Sections~\ref{sec.cover} and~\ref{sec.harm-struct}, we explain the construction of the covering space and
the definition of Dirichlet energy for the problem at hand. This material, adapted from \cite{MM2024}, is
included for completeness. As in the nearest-neighbor graph setting, the covering space depends on
the degree of the harmonic map and is constructed separately for each homotopy class.
For simplicity, Sections~\ref{sec.cover} and~\ref{sec.harm-struct} focus on simple harmonic maps,
whose degrees are given by a single integer. In Section~\ref{sec.higher}, we extend the discussion to
higher-order maps, where degrees are represented by vectors.


After that we prove the convergence of stable equilibria in the KM to a unique harmonic map within the
corresponding homotopy class. This is the goal of Section~\ref{sec.proof}, where we follow the approach
of Section~\ref{sec.gamma} to prove Theorem~\ref{thm.main-sg}, the main result of this paper.

\subsection{The covering space}\label{sec.cover} 

As in Section~\ref{sec.near}, the key step in the construction of harmonic maps  is setting up the
covering space for $K$. This space is constructed separately for each degree vector
$\bar\omega (u)$. To avoid obscuring the main ideas behind the construction with technical details
required for the general case, we focus in this subsection on the special case where
$\bar\omega (u)=(\rho_0), \; \rho_0\neq 0$, i.e., when the winding number along the boundary
of the base triangle $T$ is nontrivial and equal to $\rho_0$, while the winding numbers along all other
loops $\partial T_w, |w|>0$  are trivial.
In the following subsection, we explain the additional details needed to handle the general case.

As in Section~\ref{sec:mapsonT}, the construction of the covering space for $K$
consists of three steps: stack, cut, and identify (see Figure~\ref{fig:sg2}):
\begin{enumerate}
	\item
	First, let
	\begin{align}
		K_\times & \doteq K\times \Z, \label{G-stack}\\
		K^s & \doteq K\times \{s\} \subset K_\times, \quad s\in\Z, \label{k-sheet}\\
		K^s_i &\doteq F_i(K^s),\quad i\in S,\quad s\in \Z,\\
		K^s_1\cap K^s_2  =\{x^s\}, &\quad  K_2^s\cap K_3^s =\{y^s\}, \quad  K_3^s\cap 
		K_1^s =\{z^s\}, \quad s\in\Z, \label{xyz}.
	\end{align} 
	\begin{figure}[h]
		\centering
		\includegraphics[width = .9\textwidth]{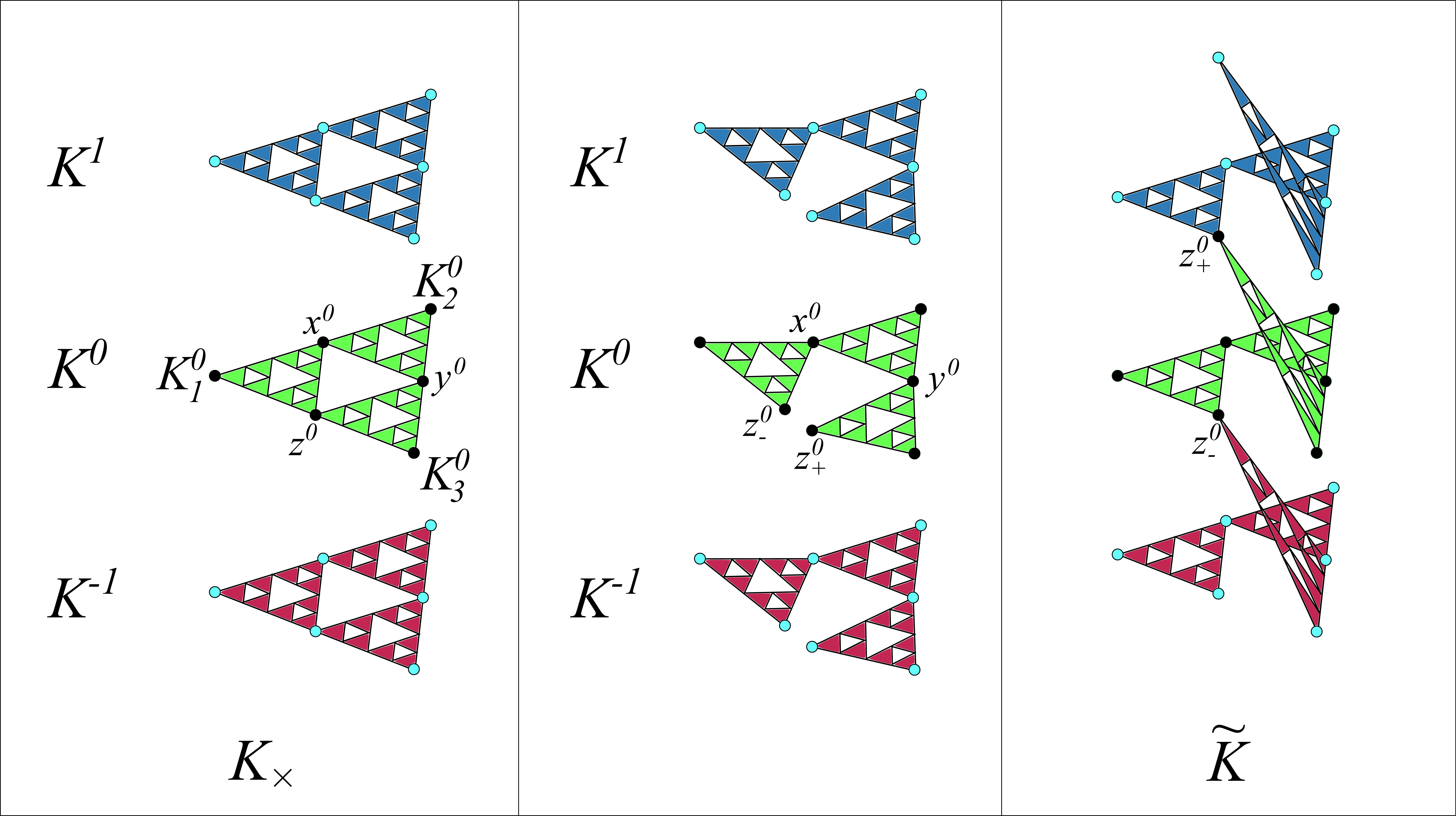}
		\caption{The construction of the covering space of $K$ corresponding to the degree
			$\bar{\omega}(u) = (1)$.}
		\label{fig:sg2}
	\end{figure}
	\item Cut each $K^s$ at $z^s$, i.e., replace $z^s$ with two distinct copies 
	\begin{equation}\label{cut-z}
		z^s_- = v^s_{1\bar{3}}\quad\mbox{and}\quad
		z^s_+ = v^s_{3\bar{1}}.
	\end{equation}
	\item Identify
	\begin{equation}\label{identify}
		z^s_+ =v_{3\bar{1}}^s \simeq v_{1\bar{3}}^{s+\rho_0} = z^{s+\rho_0}_-, \qquad s\in \Z.
	\end{equation} 
	The resultant covering space is then $\tilde K\doteq K_\times/\simeq$.
	The copies of $K$, belonging to different levels $s\in\Z$, compose the sheets of $\tilde K$.
	We keep denoting them by $K^s$. Each sheet contains both copies of $z^s:$ $z_-^s$ and $z^s_+.$
	$K^0$ is called the \textit{fundamental domain}.
\end{enumerate}

Finally, we introduce a family of graphs $\Gamma_m^s=(V_m^s,E_m^s)$ approximating $K^s$. The $\Gamma_m^s$ are constructed in the same way as the approximating graphs 
$\Gamma_m$  of SG (see Figure~\ref{fig:sg_graphs}), with the only 
distinction that $z^s$ is replaced with the two copies $z^s_-$ and $z^s_+$ (see Figure~\ref{fig:sg_graph_cut}). 
By identifying $V_m^s\ni z_+^s\simeq z_-^{s+\rho_0}\in V_m^{s+\rho_0}$, 
we obtain the discretization
of the covering space $\tilde K$, denoted $\tilde \Gamma_m, \; m\in\N$. The set of nodes
of $\tilde\Gamma_m$ is $\tilde V_m$. Then $\tilde V_\ast=\bigcup_{m=0}^\infty \tilde V_m$,
and $\tilde V_\ast^s =\bigcup_{m=1}^\infty V_m^s.$ As before, $\tilde V_\ast$ is dense in $\tilde K.$

\begin{figure}[h]
	\centering
	\includegraphics[width = .4\textwidth]{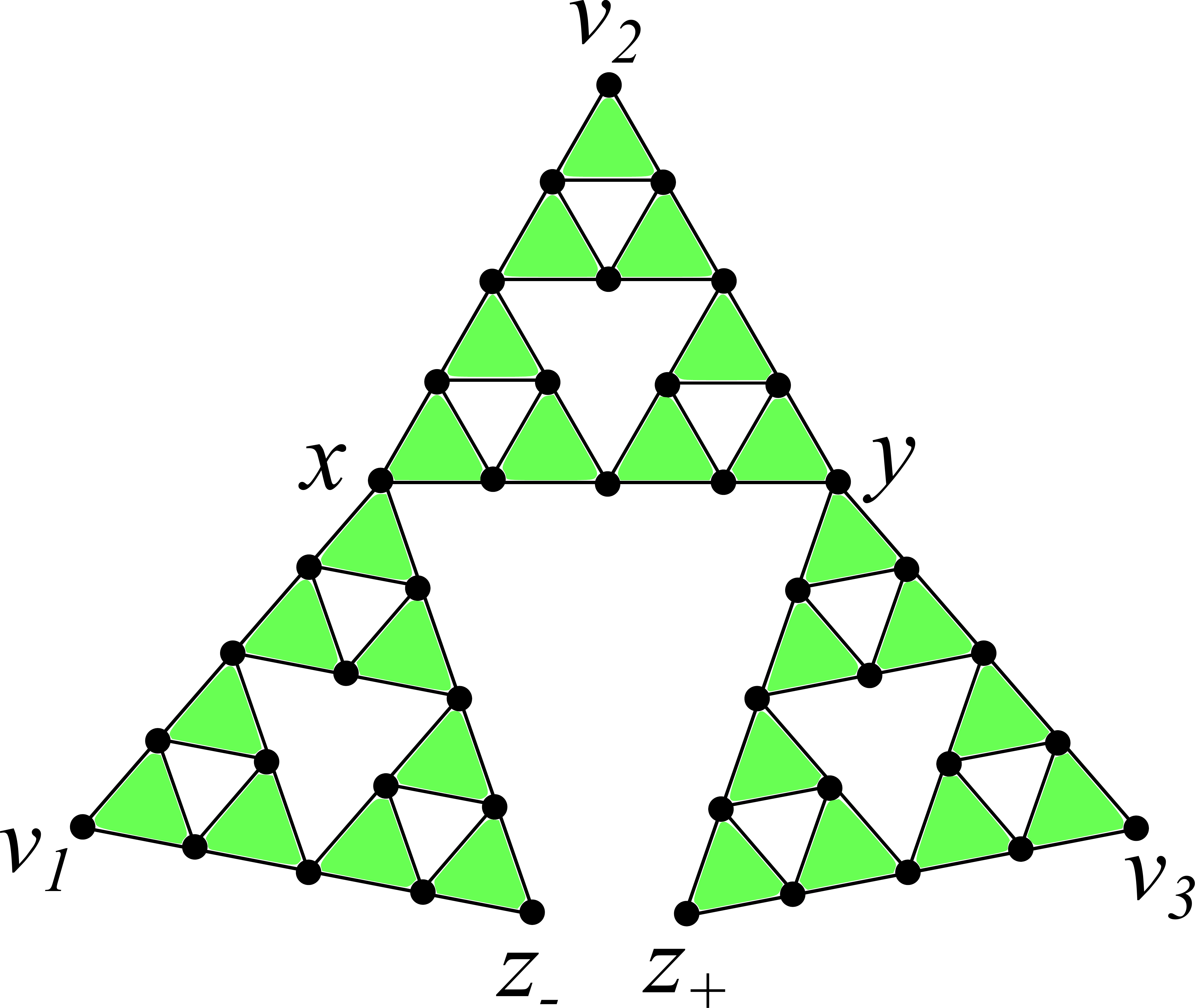}
	\caption{Approximating graphs $\Gamma_m^k$ for each sheet $K^s$ of the covering
		space $\tilde{K}$. Since our construction results in two copies of the vertex $z$,
		the resultant graphs are distinct from Figure \ref{fig:sg_graphs}. Superscripts $s$ are suppressed for simplicity.}
	\label{fig:sg_graph_cut}
\end{figure}

\subsection{Harmonic structure on $\tilde K$}\label{sec.harm-struct}
\setcounter{equation}{0}

In this subsection, we construct a harmonic map from $K$ to $\T$ without imposing
any boundary conditions. To this end, we first construct a real-valued harmonic function on the fundamental domain $K^0$, which will extend to the covering space $\tilde K$.

The energy form  on $\Gamma^0_m$ is defined as follows
\begin{equation}\label{EGamma}
	\cE_m[\Gamma^0_m]( f) = \left(\frac{5}{3}\right)^m\sum_{\{x,y\}\in E_m^0} \frac{\left( f(x)- f(y)\right)^2}{2},
	\qquad  f\in L(V_m^0,\R), m\in\N.
\end{equation}

We consider minimizing $\cE_m[\Gamma^0_m]( f)$
subject to the jump condition
\begin{equation}\label{jump}
	f(z^0_+) =  f(z^0_-)+\rho_0,
\end{equation}
As before, to remove translation invariance we fix $f(v^0_1)=0$. Since we will work exclusively on the fundamental domain $K^0$, in the following we suppress some superscripts for simplicity.

The minimization problem can be reformulated as 
\begin{equation}\label{variational}
	\cE_m[\Gamma^0_m](f) \longrightarrow\; \min_{f\in H^0_m},
\end{equation} 
where 
\begin{equation}\label{def-H}
	H^0_m\doteq \{ f\in L(V_m^0, \R):\quad f(v_1)=0,\;
	f(z_+) = f(z_-)+\rho_0\}.
\end{equation}

The energy form $\cE_m[\Gamma^0_m]$ inherits the key properties of the 
Dirichlet form $\cE_m$ (see \eqref{variational-harmonic}, \eqref{variational-extension}):
\begin{enumerate}
	\item A $\Gamma_m^0$-harmonic\footnote{
		In the present context, a function $f$ is said to be  $\Gamma_m^0$-harmonic if it satisfies the discrete Laplace equation:
		\begin{equation}
			\sum_{(x,y)\in E_m^0} (f(y)-f(x)) = 0,\qquad \forall x\in V_m^0\setminus \{v_1,z_+,z_-\}.
		\end{equation} 
	} 
	$f_m^0\in H_m^0$ minimizes $\cE_m[\Gamma^0_m]$ over $H_m^0$:
	\begin{equation}\label{E-harmonic}
		\cE_m[\Gamma^0_m](f^0_m)=
		\min\left\{ \cE_m[\Gamma^0_m](f):\; f\in H_m^0 \right\}.
	\end{equation}
	\item The minimum of the energy form $\cE_m[\Gamma^0_m]$ over all extensions of 
	$f\in H_{m-1}^0$ to $H_m^0$ is equal to $\cE_{m-1}[\Gamma^0_{m-1}](f)$:
	\begin{equation}\label{E-extension}
		\min\left\{ \cE_m[\Gamma^0_m](\tilde f):\; \tilde f\in H_m^0,\; 
		\tilde f|_{V_{m-1}}=f\in H_{m-1}^0 \right\} = \cE_{m-1}[\Gamma^0_{m-1}](f).
	\end{equation} 
	Moreover, this minimizer is obtained by taking the harmonic extension of $f$.
\end{enumerate}

As in \eqref{monotone-conv}, since energies are monotone increasing, the limit
\begin{equation}\label{eq:limit_k0_energy}
	\cE_{K^0}(f) = \lim_{m\to\infty} \cE_m[\Gamma_m^0](f)	
\end{equation}
is well defined for $f\in C(K^0,\R)$.

Applying the arguments in \cite[Lemma~5.1]{MM2024}, we also conclude that the variational problem \eqref{variational} has a unique solution
$f^0_m \in H^0_m$:
$$
\cE_m[\Gamma_m^0](f^0_m)=\min_{f\in H_m^0} \cE_m[\Gamma_m^0](f),
$$
and the minimizers $f^0_m$ satisfy the following properties:
\begin{enumerate}
	\item $f_m^0$ \; is $\Gamma^0_m$-harmonic,
	\item $f^0_{m+1}|_{V^0_m} = f_m^0$. 
\end{enumerate}

\begin{remark}
	This second property requires comment. In principle, it is not obvious that minimizers must agree on common vertices (particularly as boundary conditions at $v_2$ and $v_3$ are not fixed, in contrast to \cite{MM2024}). However, given $f_m^0$,  from \eqref{E-extension}, its harmonic extension $\tilde{f}_m^0 \in H_{m+1}^0$ satisfies $\cE_{m+1}[\Gamma_{m+1}^0](f^0_{m+1}) = \cE_{m+1}[\Gamma_{m+1}^0](\tilde{f}_m^0)$. Since minimizers are unique, it must follow that $\tilde{f}_m^0 = f^0_{m+1}$ and so $f^0_{m+1}|_{V^0_m}= f_m^0$.
\end{remark}

\begin{figure}[h]
	\centering
	\includegraphics[width = .3\textwidth]{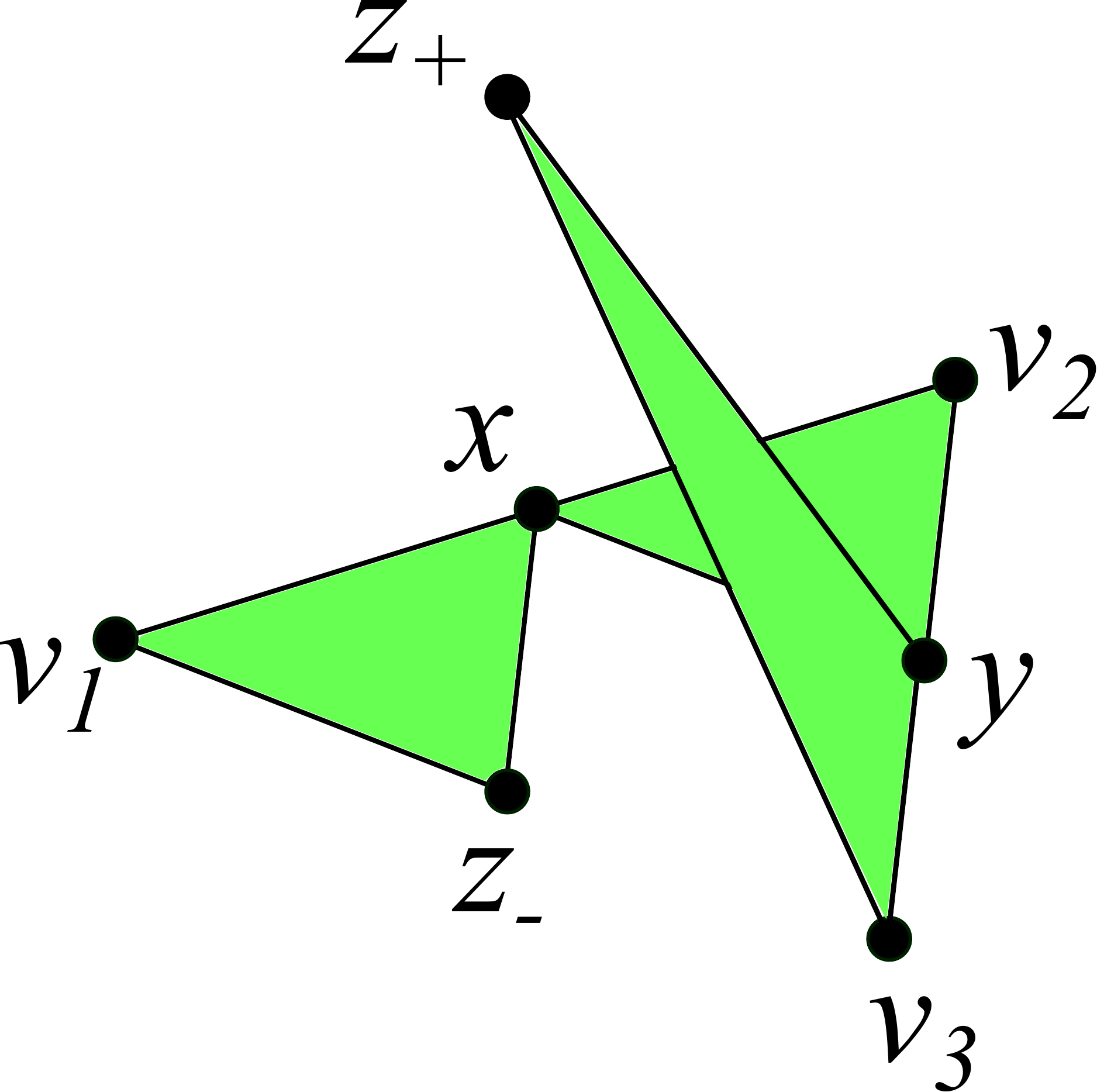}
	\caption{The boundary conditions for the minimization problem \eqref{variational}:
		the values of the function is prescribed only at  $v_1$ and the  jump condition
		is imposed at $z_-$ and $z_+$. }
	\label{fig:sg_fund_domain}
\end{figure}
%

Thus, minimizers of \eqref{variational} are constructed by first minimizing $\cE_1[\Gamma_1^0]$, followed by repeated application of the harmonic extension.
Let $\hat f_m$ be  the minimizer of $\cE_m[\Gamma_m^0](f)$. By \eqref{E-extension}, 
$\cE_m[\Gamma_m^0](\hat f_m)=\cE_1[\Gamma_1^0](\hat f_1)$ and $\hat f_m|_{V_1} = \hat f_1$. 

So, we first find the values of $\hat f_1$ at points $x,y,z:=z_+, v_2,v_3$, by minimizing  $\cE_1[\Gamma_1^0]$ (see Fig.~\ref{fig:sg_fund_domain}). Direct calculation yields the unique solution
\begin{equation*}
	\begin{pmatrix} 
		\hat f_1(x)\\\hat f_1(v_2)\\\hat f_1(y)\\\hat f_1(v_3)\\\hat f_1(z)
	\end{pmatrix} 
	=
	\begin{pmatrix} 1\\1\\1\\1\\1\end{pmatrix} \hat f_1(v_1)+
	\begin{pmatrix}
		1/6\\2/6\\3/6\\4/6\\5/6
	\end{pmatrix}
	\rho_0.
\end{equation*}
Using the harmonic extension \eqref{classical-extension} in the 
subdomains $v_1xz_-$, $xv_2y$, and $z_+yv_3$ inductively, we obtain a minimizer $\hat f^*$,  on $V_*^0$. 

%
%

We extend it by continuity to the fundamental domain $K^0$,
and further to a uniformly
continuous function on the covering space:
\begin{equation}\label{k-period}
	\mathbf{\hat f}^\ast (x,s) = \hat f^\ast(x) +s , \qquad x\in K^0, \; s\in \Z.
\end{equation}

Recall that $\tilde\Gamma_m$ and $\tilde V_m$ stand for the family of graphs approximating
the covering space $\tilde K$ and their vertex sets respectively (see \S\ref{sec.cover}).
By \cite[Theorem~5.2]{MM2024}, the restriction of $  \mathbf{\hat f}^\ast$ on $\tilde V_m$, $ \mathbf{\hat f}^\ast|_{\tilde V_m}$,
is  $\tilde\Gamma_m$-harmonic for every $m\in \N$, {including at the jump vertices $z^s_-$ and $z_+^s$},
i.e., it satisfies the discrete Laplace equation
\begin{equation}\label{G-harmonic}
	\Delta_m   \mathbf{\hat f}^\ast|_{\tilde V_m}(x) = 0\quad \forall x\in \tilde{V}_m\setminus \{v^s_1\}.   
\end{equation}

Finally, by restricting the domain of $  \mathbf{\hat f}^\ast$ to the fundamental domain 
and by projecting the range of $  \mathbf{\hat f}^\ast$ to $\T$, we obtain the harmonic map from
the SG to the circle $f^\ast:\, K\to \T$:
\begin{equation}\label{project-back}
	f^\ast\doteq   \mathbf{\hat f}^\ast|_{K^0}\mod 1.
\end{equation}

H\"older continuity of functions on $K^0$ also holds with minor adjustments. 
The fundamental domain $K^0$ is the union of the three subdomains
$K^0_i, \; i=1,2,3$ (see Fig.~\ref{fig:sg_graph_cut}), each a copy of $K$ itself. Specifically,  $K^0_1, K^0_2,$ and $K^0_3$
correspond to the intersection of $K^0$ with the closed triangles $v_1xz_-,$ $xv_2y$,
and $z_+yv_3$, respectively.  Let $ f\in \operatorname{dom}(\cE_{K^0})$.
By Lemma \ref{lem.SG_holder}, we immediately see that the restriction of
$ f$ on each of the subdomains $K^0_i, i\in [3]$ is
a H\"older continuous function. This yields the following lemma.
\begin{lemma}\label{lem.Holder-G^0}
	Let $ f\in \operatorname{dom}(\cE_{K^0})$ then for every $K_i^{0}, i\in [3]$ and $x,y\in K^0_i$,
	\begin{equation*}
		\frac{| f(x)- f(y)|}{|x-y|^\beta}\leq C\sqrt{\cE_{K^0}( f)},\quad
		\beta = \frac{\log(5/3)}{2\log(2)},
	\end{equation*}
	where $C>0$ is independent of $ f$. 
\end{lemma}

We conclude this subsection by showing that minimization of the Dirichlet energy on the fundamental
domain without imposing any boundary conditions on $V_0$ results in a solution satisfying
zero Neumann boundary conditions.

\begin{lemma}\label{lem.natural}
	$f^\ast$ satisfies natural boundary conditions \eqref{HM-neu}.
\end{lemma}
\begin{proof}
	We recall that for a harmonic function $\hat f^* \in C(K^0,\R)$, the normal derivative at a boundary vertex
	$v_j \in V_0$ is 
	\begin{equation*}
		\partial_{\bf n} \hat f^*(v_j) = \left(\frac{5}{3} \right)^m
		\sum_{y:\,y\sim_m j} \left( \hat f^*(y)-\hat f^*(v_j) \right),
	\end{equation*}
	which is well-defined because the right hand side is constant in $m$ for harmonic functions \cite{Str06}. In the following, we work always on the fundamental domain and suppress superscripts/subscripts for simplicity.
	
	For every $m\in \N$, the Euler-Lagrange equations for \eqref{EGamma} at the boundary points $v_2$ and $v_3$ is 
	\begin{equation}\label{eq:EL_lift}
		\sum_{j:\,j\sim_m i} \left(\hat f^*(v_j)-\hat f^*(v_i)\right)=0, \quad i=2,3,
	\end{equation}
	from which it follows $\partial_{\bf n} \hat f^*(v_2)=\partial_{\bf n}\hat f^*(v_3) = 0$. To see that $\partial_{\bf n} \hat f^*(v_1)=0$ as well, we use a discrete divergence theorem. For $x\in \Gamma^0_n\setminus V^0_0$, define 
	\begin{equation*} 
		\Delta_n \hat {f}^*(x) = \left\{ \begin{array}{cc} 
			\sum_{y\sim_n x} (\hat f^*(y)-\hat f^*(x)), &\mbox{ if } x \mbox{ is not a jump vertex}\\
			\sum_{y\sim_n x_+} (\hat f^*(y)-\hat f^*(x_+))+\sum_{y\sim_n x_-} (\hat f^*(y)-\hat f^*(x_-)), &\mbox{ if } x \mbox{ is a jump vertex}.
		\end{array}\right.
	\end{equation*} 
	This way,  $\Delta_n \hat f^*(x)=0$ for all $x\not\in V_0$ (cf. \eqref{G-harmonic}). Then, summing over all junction points results in the discrete divergence theorem: 
	\begin{equation*} \sum_{x\in V_n\setminus V_0} \Delta_n \hat f^*(x)=\sum_{v_i\in V_0}\sum_{y:\,y\sim_n v_i} \hat f^*(y)-\hat f^*(v_i).
	\end{equation*} 
	Together with \eqref{eq:EL_lift} this shows $\partial_{\bf n}\hat  f^*(v_1)=0$. Projecting  $f^*=\hat f^*|_\T$ preserves normal derivatives, so the result follows.
\end{proof}

\subsection{Higher order harmonic maps}\label{sec.higher}

In the previous section, we explained how to construct a harmonic map when the
winding vector has a single nontrivial entry, $\bar{\omega}(f) = (\rho_0)$. The procedure
for a general winding vector, $\bar{\omega}(f) = (\rho_0, \rho_1, \dots, \rho_k)$, follows the
same general approach, but it requires additional cuts to construct the covering space.
The idea is to introduce one cut point for each nonzero entry in the winding vector.
Care must be taken in choosing these points, as the loops used for cutting must be
linearly independent. The choices of cut points when $\bar{\omega}(f) = (\rho_0, \rho_1, \rho_2, \rho_3)$
and $\bar{\omega}(f) = (\rho_0, \rho_1, \rho_2, \dots, \rho_{12})$ are shown
in Figure~\ref{fig:gen_cuts}. The general procedure of selecting cut points for a p.c.f. fractal is
explained in detail in Sections 7 and 8 of \cite{MM2024}. For completeness, 
we discuss the construction of the covering space for SG for an arbitrary winding vector below
and refer the reader to \cite{MM2024} for further details.

Let $f\in C(K,\T)$ be a harmonic map of order $N\in \N$, i.e., $\bar{\omega}(f) = (\rho_0,\rho_1,\dots, \rho_k)$
with
$$
\frac{3^N-1}{2} \le k< \frac{3^{N+1}-1}{2}.
$$
Note that to each nonzero entry of $\bar{\omega}(f)$, $\rho_i, i\le \frac{3^{N+1}-3}{2}$,
corresponds a unique loop $\partial T_w, w\in S^N$. For the remainder of this subsection,
let us denote $\rho_i$ by $\rho_w$, i.e., by the itinerary of the cell, whose boundary has
winding number $\rho_i$.

For each nonzero $\rho_w$ with $n:=|w|\le N$, we cut $K^s$ at a vertex  $z_w$
satisfying $z_w \in \partial T_w \setminus (\bigcup_{|v|=n-1} \partial T_v),$
see Fig. \ref{fig:gen_cuts}.
Denote the two copies of $z_w$ by $(z_w)_-$ and $(z_w)_+$.
Finally, identify vertices in the sheets $K^s$ by
\begin{equation*}
	(z_w^k)_+ \simeq (z_w^{k+\rho_w})_-.
\end{equation*}
The covering space is then $\tilde{K} \doteq K_{\times}/ \simeq$,
and the fundamental domain $K^0$. Construction of the approximating graphs
of the covering space, $\Gamma^s_m=(V^s_m,E^s_n)$, follows analogously.

\begin{figure}[h]
	\centering
	\includegraphics[width = .3\textwidth]{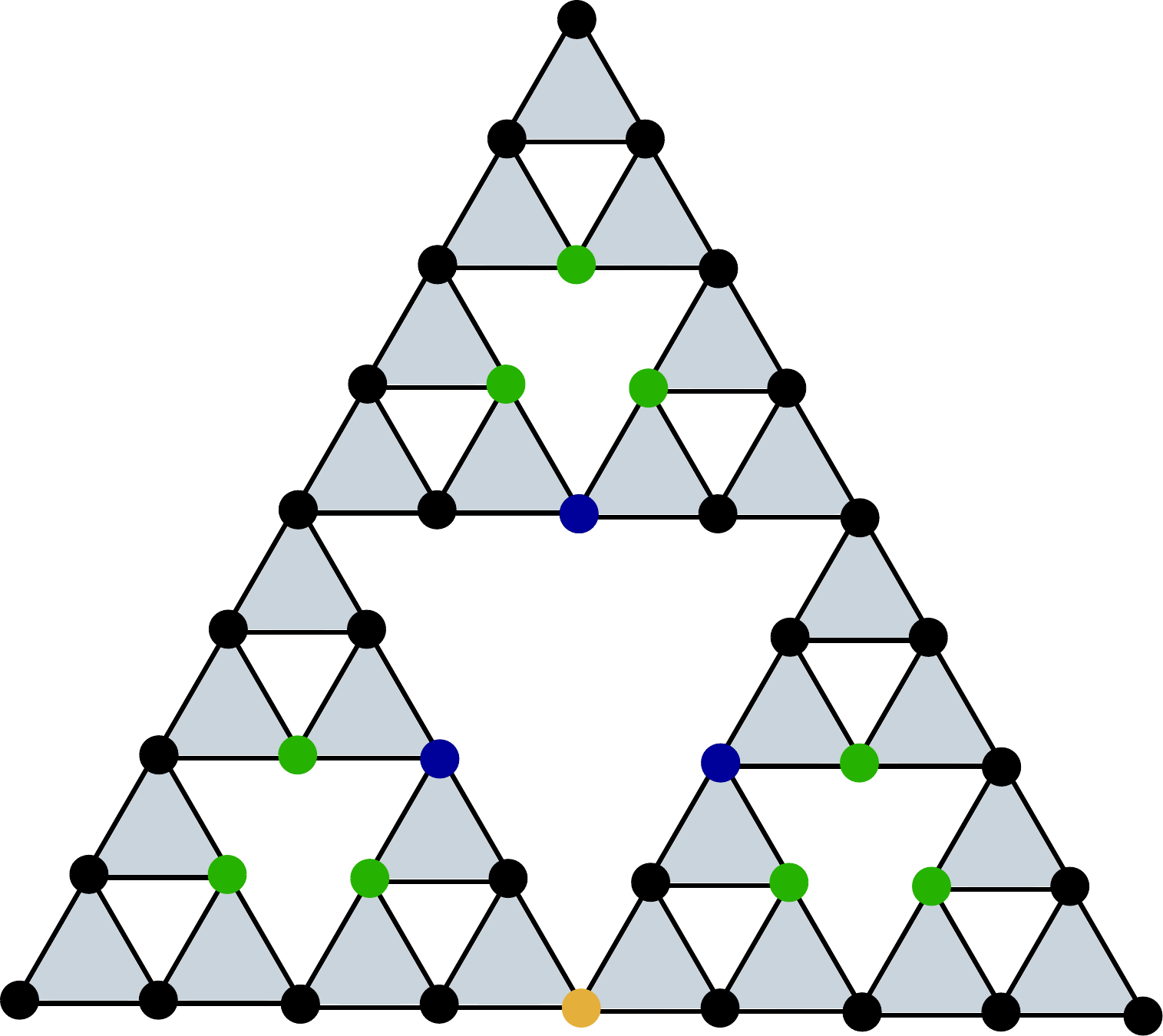}
	\caption{For arbitrary degree $\bar{\omega}(f)$, each non-zero entry requires an
		associated cut vertex. As before, for $\partial T$ we cut at the single orange vertex.
		For loops $\partial T_1,\partial T_2,\partial T_3$ we cut the three blue vertices.
		Finally, for loops $\partial T_{ij}$ there are 9 corresponding green vertices.}
	\label{fig:gen_cuts}
\end{figure}

Constructing the harmonic structure on $\tilde{K}$ also follows along the same lines as before.
We introduce the
corresponding spaces on the fundamental domain $K^0$:
\begin{align*}
	H^0_m = \{  f \in L(V^0_m,\R)\colon  & f(v^0_1)=0,\\
	&\;  f((z^0_w)_+)= f((z^0_w)_-)+\rho_w, \forall w\in S^n, |n|\leq m\}
\end{align*}
and energies $\cE_m[\Gamma_m^0]$ as in \eqref{EGamma}.
The resultant harmonic structure satisfies the same results as in
Section~\ref{sec.harm-struct}, so we construct a harmonic function ${\bf \hat f}^*$ on
the covering space by minimizing the $\cE_m[\Gamma_m^0]$ on the fundamental domain within the class
$H_m^0$, followed by inductive application of the harmonic extension algorithm.
The projection of the range of ${\bf \hat f}^*$ to $\T$ yields the desired HM,
$f^*$, with degree $\bar{\omega}$.

Arguing as in \S~\ref{sec.harm-struct}, we represent the fundamental domain as a union of
closed $m$-cells $G^0=\bigcup_{|w|=m} K^0_w$. Recall $K^0_w=F_w(K^0)$.  As before,
the restriction of $ f\in\operatorname{dom}(\cE_{K^0})$ on each of these subdomains
is a H\"older continuous function.
\begin{lemma}\label{lem:holder_sg_lift}
	Let $ f\in \operatorname{dom}(\cE_{K^0})$ and  $K^0_w,$ $|w|=m$. For any 
	$x,y\in K^0_w$ we have
	\begin{equation*}
		\frac{| f(x)- f(y)|}{|x-y|^\beta}\leq C\sqrt{\cE_{K^0}( f)}, \quad \beta = \frac{\log(5/3)}{2\log(2)},
	\end{equation*}
	where $C$ is independent of $ f$. 
\end{lemma}

\section{The proof of Theorem~\ref{thm.main-sg}}\label{sec.proof}
\setcounter{equation}{0}

We are now prepared to prove Theorem~\ref{thm.main-sg}. We work with the
lifts of solutions of the KM restricted to the fundamental domain. The proof for the most
part coincides with the proofs of Theorems~\ref{thm:SGThm} and \ref{thm:conv_minSG}.
Only minor adjustments are required to deal with the fact that the covering space is a union
of finitely many copies of $K$. 
Below, we comment on the necessary
adjustments to the arguments in Section~\ref{sec.gamma}.

Let
$$
\omega^\ast=\left(\rho_0, \rho_1, \dots,\rho_{M(m)}\right)\in\Z^{M(m)+1}, M(m)\doteq \frac{3(3^m-1)}{2},
$$
be arbitrary but fixed. We assume that
$$
\sum_{i=M(m-1)+1}^{M(m)} |\rho_i|>0.
$$
The last condition means that at least some cells of order $m$ are used in the construction of the
covering space corresponding to $\omega^\ast$. 

Next, construct the covering space using the algorithm explained in
\S\S\ref{sec.harm-struct}, \ref{sec.higher}. Decompose the fundamental domain $K^0$ into the union of
$m$-cells, $K^0=\bigcup_{|w|=m} K^0_w$. Let $\hat u_m$ stand for the unique minimizer of the
Dirichlet energy
$
\cE_m[\Gamma_m^0]
$
over $H^0_m$ and denote by $\hat u^*$ the harmonic extension of $\hat u_m$ to $K^0$. Finally, $u^\ast$ is obtained
by projecting the co-domain of $\hat u^*$ to $\T$. By construction (and using Lemma \ref{lem.natural}), $u^\ast$ satisfies \eqref{HM-Lap} - \eqref{HM-deg}.

\noindent{\bf $\Gamma$-convergence.} We first establish $\Gamma$-convergence of the Kuramoto energies on $K^0$.
{Take 
	\begin{align*}
		X = \{  f \in C(K^0,\R)\colon  & f(v^0_1)=0,\\
		&\;  f((z^0_w)_+)= f((z^0_w)_-)+\rho_w, \forall w\in S^n, |n|\leq m\},
	\end{align*}
	equipped with the supremum norm.}
As before, $\Glim{n} \cE_n = \cE_{K^0}$ by virtue of continuity of the $\cE_n$ and monotone pointwise convergence of the energies. Next define the Kuramoto energy on $\Gamma_n^0$:
\begin{equation}\label{eq:k0_km_energy}
	\cJ_n(u) = \left(\frac{5}{3}\right)^n\frac{1}{4\pi^2} \sum_{\{i,j\}\in E_n^0} (1-\cos(2\pi(u(v_j)-u(v_i)))).
\end{equation} 
 Repeating the proof of Theorem \ref{thm:SGThm} establishes $\Glim{n} \cJ_n = \cE_{K^0}$. 

\noindent{\bf Existence of sequence of stable equilibria.}  It remains to show that there is a sequence of stable steady states of the KM,     $u^{n}\in L(V_n^0,\T),$ converging to $u^\ast$. 

We proceed as in Theorem \ref{thm:conv_minSG} to construct minimizers of $\cJ_n$ converging to $u^*$.
Importantly, since we now seek minimizers in the space $H_m^0$
(which does not assume Dirichlet boundary conditions), then minimizers of $\cJ_n$ correspond to stable equilibria of \eqref{KM-intro} (cf. \eqref{min-intro}). Since the proofs are essentially the same, we only outline the main steps.

Recall that
$K^0 = \bigcup_{|w|=m} K^{0}_w$. Introduce the space of H\"older continuous functions on each
$K^{0}_w$ with norm
\begin{equation*}
	\|f\|_{C^\beta(K^{0}_w)} =  \sup_{x\in K^{0}_w} |f(x)|+  \sup_{x\neq y} \frac{|f(y)-f(x)|}{|x-y|^\beta}.
\end{equation*}
Let $C^\beta(K^0)$ be the space of functions on $K^0$ that are H\"older continuous on each $K^0_w$ with exponent $\beta = \frac{\log(5/3)}{2\log 2}$. 
Fix $\varepsilon>0$ and consider the minimization of $\cJ_n$ on the closed set
\begin{equation*}
	\overline{B_\varepsilon(u^*)}:=\{f\in C^0(K^0){\cap X}\colon \max_w \|f-u^*\|_{C^0(K^0_w)}\leq \varepsilon\}.
\end{equation*}
Restricting the minimization problem to the space of harmonic splines, $X^n$, (cf. \eqref{harm-splines}) is a finite dimensional problem, and so repeating the compactness argument of $\cJ_n$ guarantees a minimum on $X^n\cap \overline{B_\varepsilon(u^*)}$.

We need only show that minimizers are attained on the interior of the ball $B_\varepsilon(u^*)$ for all $n$ sufficiently large (so as to be true local minima). Suppose not. Then, there is a sequence $f^n\in C^0(K^0){\cap X}$
with $\max_w \|f^n-u^*\|_{C^0(K^0_w)}= \varepsilon$.  Restricting to $K^0_1$, repeat the proof of Theorem \ref{thm:conv_minSG}: uniform equicontinuity of the $f^n|_{K^0_1}$ implies a bound on the Dirichlet energies. Applying Lemma \ref{lem:holder_sg_lift} implies uniform boundedness of the sequence in $C^\beta(K^0_1)$. Thus by compactness of the embedding, we deduce that there is is a subsequence $f^{n_k}\to \bar{f}\in C(K^0_1){\cap X}$. If necessary, pass to a further subsequence so that
$f^{n_{k_\ell}}\to \bar{f}\in C(K^0_1\cup K^0_2){\cap X}$. As there are finitely many subdomains $K^0_w$, repeating
this results in a subsequence (which we denote $f^{n_j}$ for simplicity) that converges on the entire domain:
$f^{n_j}\to \bar{f} \in{ X\subset}C(K^0)$.

As before, by the $\Gamma$-convergence of $\cJ_n$ to $\cE_{K^0}$, we obtain the contradiction
$\bar{f}=u^*$. Thus, for all $n$ sufficiently large, there is a (local) minimizer, $u^n$, of $\cJ_n$ in the interior of $B_\varepsilon(u^*)$. By construction $u^n$ has degree $\omega^\ast$. Thus \eqref{harm-spline-estimate} follows.  Following the same strategy as Corollary \ref{cor.upgrade}, we can upgrade to convergence in $C^\beta$.


\section{The KM on post-critically finite fractals}\label{sec.pcf}
  \setcounter{equation}{0}
The analysis of functions on SG naturally extends to a broader class of self-similar domains
known as p.c.f. fractals. Examples include higher-order Sierpinski Gaskets, the pentagasket, and the hexagasket (see Fig.~\ref{fig:pcf_figs}), to name a few.  
In this final section, we describe how the preceding results can be
extended to the KM on p.c.f. fractals. While there are some technical differences in the construction of covering spaces (highlighted below), the techniques developed above remain applicable.
\begin{figure}[h]
\centering
\includegraphics[width=.3\textwidth]{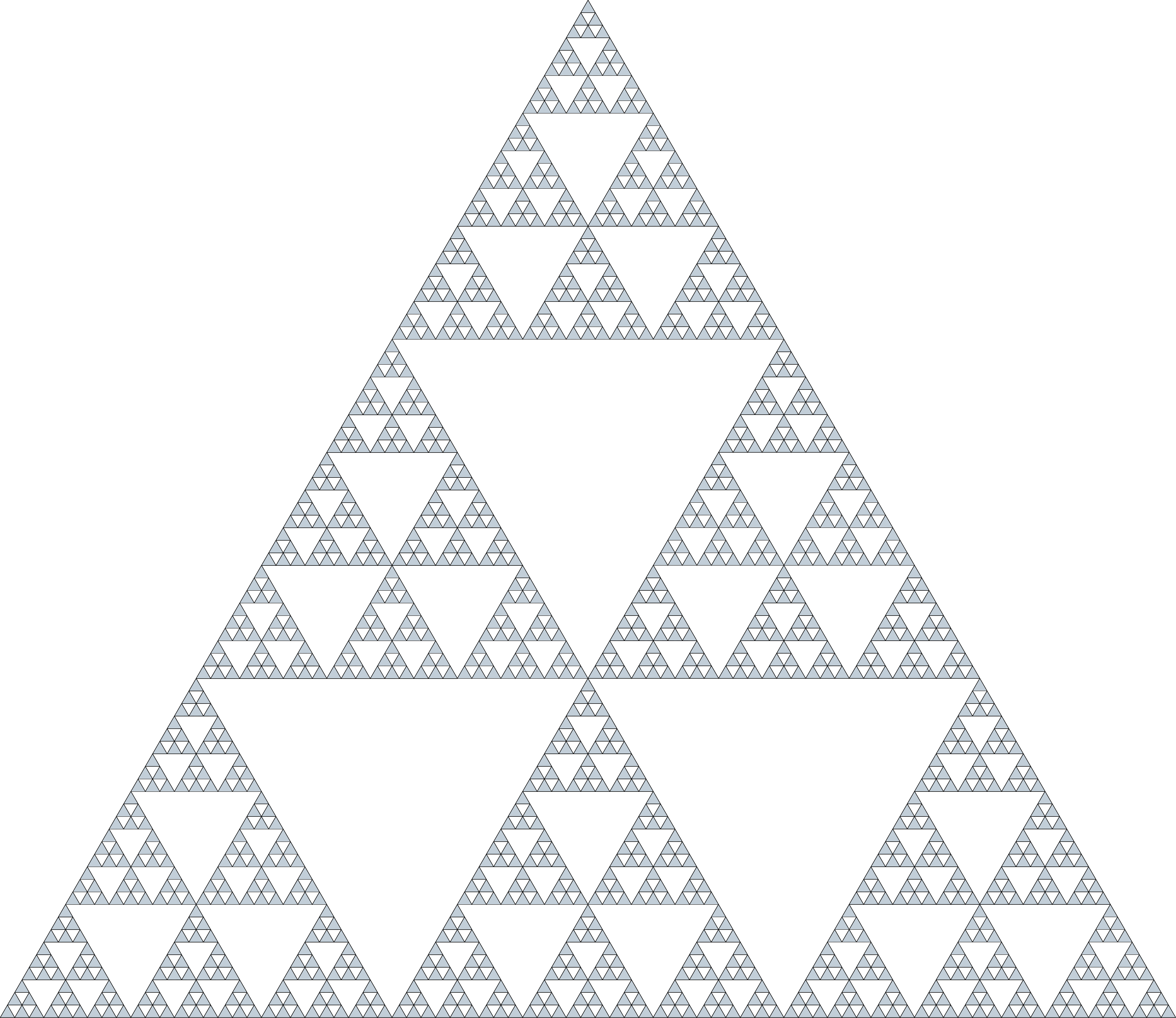}
\includegraphics[width=.3\textwidth]{figs/pentagasket.pdf}
\hspace{4mm}
\includegraphics[width=.25\textwidth]{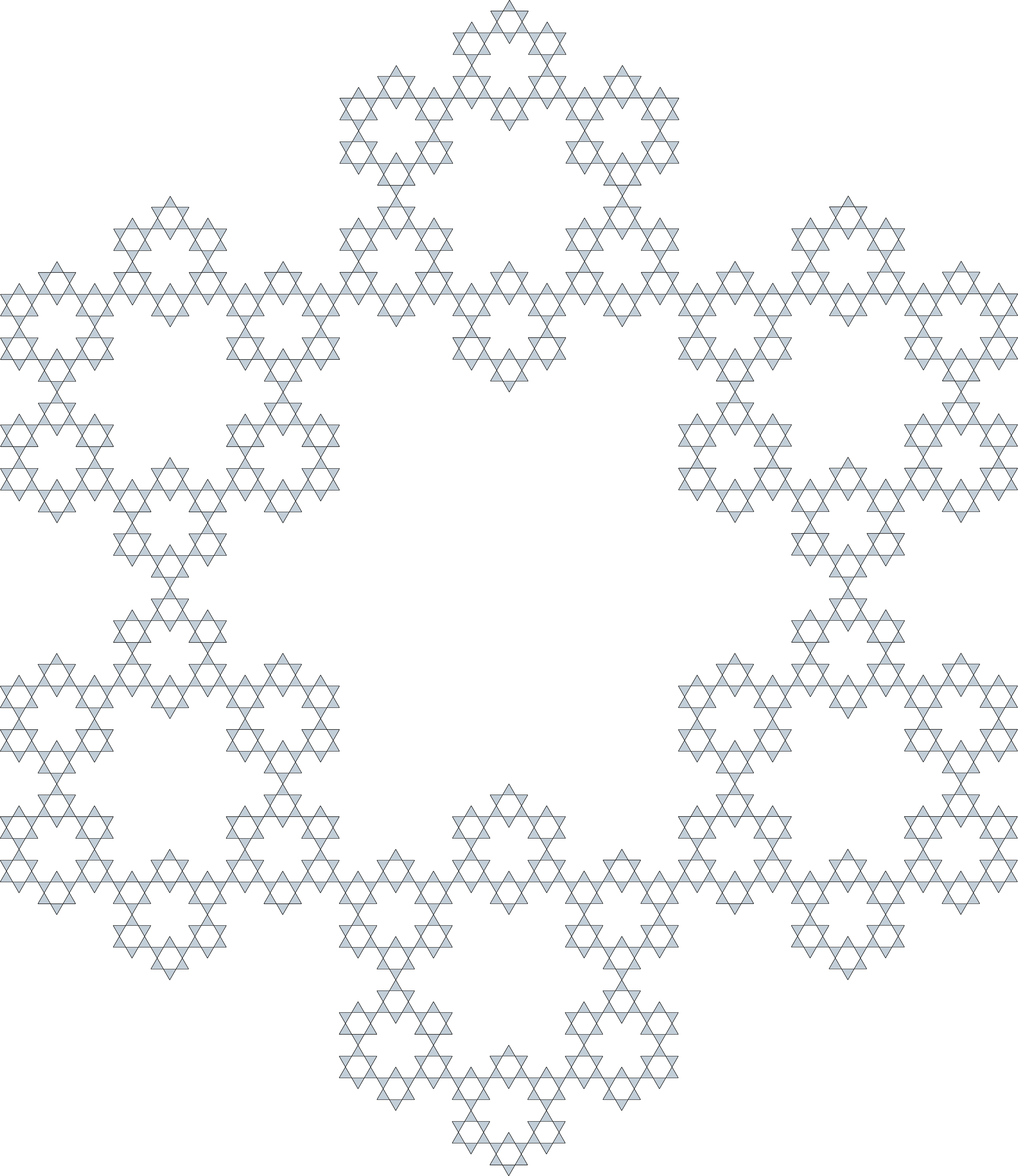}
\caption{Representative examples of p.c.f. fractals: the three-level SG, $SG_3$, the pentagasket, and the hexagasket. For details on their construction see \cite{Str06,MM2024}.}
\label{fig:pcf_figs}
\end{figure}

  \subsection{Background on p.c.f. fractals}
  
  We first review some relevant background of p.c.f. fractals and
  refer the reader to \cite{Kig01} for further details.
  Let $K\subset \R^d$ be a non-empty, connected, and compact set satisfying
  \begin{equation}
  	K=\bigcup_{i=1}^N F_i(K),
  \end{equation}
  where the $F_i:\R^d\to \R^d$ are contractive similitudes
  $$
  |F_i(x)-F_i(y)|=\lambda_i|x-y|, \; 0<\lambda_i<1,\; 1\leq i\leq N.
  $$

  Denoting the $n$-cells
  \begin{equation*}
  	K_w = F_w(K),\quad w=(w_1,w_2,\dots,w_n)\in [N]^n,
  \end{equation*}
  the critical set is defined by
  \begin{equation*}
  	\mathcal{C} = \bigcup_{i\neq j} (K_i\cap K_j).
  \end{equation*}
  If $\mathcal{C}\neq \emptyset$, then the post-critical set is
  \begin{equation*}
  	V_0 = \bigcup_{m\geq 1}\bigcup_{|w|=m} F_w^{-1}(\mathcal{C}).
  \end{equation*}
  If $V_0$ is finite, then $K$ is called a p.c.f. fractal. Moreover, the post-critical set $V_0$ is taken to be the boundary of $K$ (e.g., for the associated Dirichlet problem).

  We assume that there is a harmonic structure defined on $K$, which we describe below. Define the vertex sets
  \begin{equation*}
  	V_n = \bigcup_{i=1}^N F_i(V_{n-1}),\quad n\geq 1,
  \end{equation*}
  with $V_*=\bigcup_{n\geq 1} V_n$. Following \eqref{order-nodes}, denote the vertices $V_n = \{v_1,v_2,v_3,\dots\}$. Define the Dirichlet energies
  \begin{equation}
  	\label{eq:quadratic_energy_pcf}
  	\mathcal{E}_n(u) = \sum_{i,j=1}^{|V_n|}c_{ij}^n \frac{(u(v_j)-u(v_i))^2}{2},\quad u\in L(V_n,\R),
  \end{equation}
  where $c_{ij}^n\geq 0$. As before, it is possible to extend the domain of $\mathcal{E}_n$ to all $u\in C(K,\R)$.
  We assume that the $\mathcal{E}_n$ satisfy the following two conditions. 
  \begin{enumerate}
  	\item There are $0<r_i<1$ so that
  	\begin{equation}
  		\label{eq:compatability0}
  		\mathcal{E}_n(u)=\sum_{i=1}^N r_i^{-1} \mathcal{E}_{n-1}(u\circ F_i),\quad u\in L(V_n,\R).
  	\end{equation}
  	\item For every $n\geq 1$,
  	\begin{equation}
  		\label{eq:compatability}
  		\mathcal{E}_{n-1}(u) = \min\{\mathcal{E}_n(\bar{u})\colon \bar{u}\in L(V_n,\R), \bar{u}|_{V_{n-1}}=u\}.
  	\end{equation}
  \end{enumerate}
  
  In particular, \eqref{eq:compatability} implies that $\mathcal{E}_n(u)$ is non-decreasing for any $u\in C(K,\R)$ so the limit
  \begin{equation*}
  	\mathcal{E}(u)=\lim_{n\to\infty} \mathcal{E}_n(u)
  \end{equation*}
  is well-defined. Moreover, it is possible to define a harmonic
  extension algorithm analogous to \eqref{classical-extension}
  \cite{Str06}.

  As before, we write
  \begin{equation*}
  	\operatorname{dom}(\mathcal{E})=\{u\in C(K,\R)\colon \mathcal{E}(u)<\infty\}.
  \end{equation*} 
  The energies \eqref{eq:quadratic_energy_pcf} also naturally induce graph structures $\Gamma_n = (V_n,E_n)$. Namely, for all $v_j,v_i\in V_n$,
  \begin{equation*}
  	i\sim_n j \mbox{ iff } c_{ij}^n>0.
  \end{equation*}

  {
    To formulate the boundary value problem analogous to \eqref{HM-Lap}-\eqref{HM-deg}, we need to define the degree
    of the map for a p.c.f. fractal. For SG, we had a family of cycles
    $\{\gamma_0, \gamma_1, \gamma_2,\dots\}$, from which we could build the bases for the cycle spaces of $\Gamma_m$
    for every $m\in\N$. For a general p.c.f. fractal there is no such family and we have to take a different approach.
    Specifically, for a given $m\in\N$, we pick a basis of the cycle space $\cyc{\Gamma_m}$
   \begin{equation*}
  	\mathcal{B}_m=\left\{\gamma_1^{(m)},\dots,\gamma_{n_m}^{(m)}\right\}
  \end{equation*} 
and define the degree of $f\in C(K,\T)$ with respect to this basis
  \begin{equation*}
  	\bar{\omega}^{(m)}(f)=\left( \omega_{\gamma_1^{(m)}}(f),\dots, 	\omega_{\gamma_{n_m}^{(m)}}(f)\right).
  \end{equation*}
 Here, $\omega_{\gamma}(f)$ stands for  the degree of $f|_\gamma\in C(\T, \T)$. 

 The final step, which we need to discuss before we can move to constructing the covering space, is the choice
 of the cut point, $\xi^{(m)}_i$, for a given $\gamma_i^{(m)}, \in [n_m]$.  For SG, we formulated a simple rule that gave cut points for any cycle $\gamma_i, i=0,1,2,\dots$.
 For most standard examples of p.c.f. (e.g., $3$-level SG, hexagasket, pentagasket, etc) it is easy to see how to choose
 cut points in practice. However, for a general p.c.f. fractal existence of cut points relies on a set of technical assumptions
 on the structure of the cycle spaces. These assumptions are discussed in detail in \cite{MM2024} and we refer
 the interested reader to Section~8 of that paper for details.
 }

 We are now prepared to construct a covering space $\tilde{K}$ using the same procedure as in
 Section \ref{sec.hmaps}.
 Indeed, first fix a degree
\begin{equation}
	\bar{\omega}^{(m)}(f) = (\rho_1,\dots, \rho_{n_m})\in \Z^{n_m}.
\end{equation}
Let $K_\times = K\times \Z$. On each sheet $K^s = K\times \{s\}$, create two copies of the cut
vertices $z_i^{s}:= \xi_i^{(m)}$, named $(z_{i}^{s})_-$ and $(z_{i}^{s})_+$. Identifying 
\begin{equation*}
	(z_{i}^s)_+ \simeq (z_{i}^{s+\rho_i})_-,
\end{equation*}
results in the covering space $\tilde{K} = K_\times/\simeq$ with fundamental domain $K^0$.
Additionally, $K^0$ is approximated by graphs $\Gamma_m^s = (V_m^s,E_n^s)$. 

As a visual example, Figure \ref{fig:sg3_cut} shows this construction on $\Gamma_1$ of $SG_3$,
a higher-order Sierpinski Gasket (see, e.g., \cite{MM2024} for a precise construction of this higher
order SG). Given degree $\bar{\omega}^{(m)}$, gluing cut points across sheets results in the associated
covering space $\tilde{K}$.

\begin{figure}[h]
	\centering
	{\bf a)}\qquad\includegraphics[width = .4\textwidth]{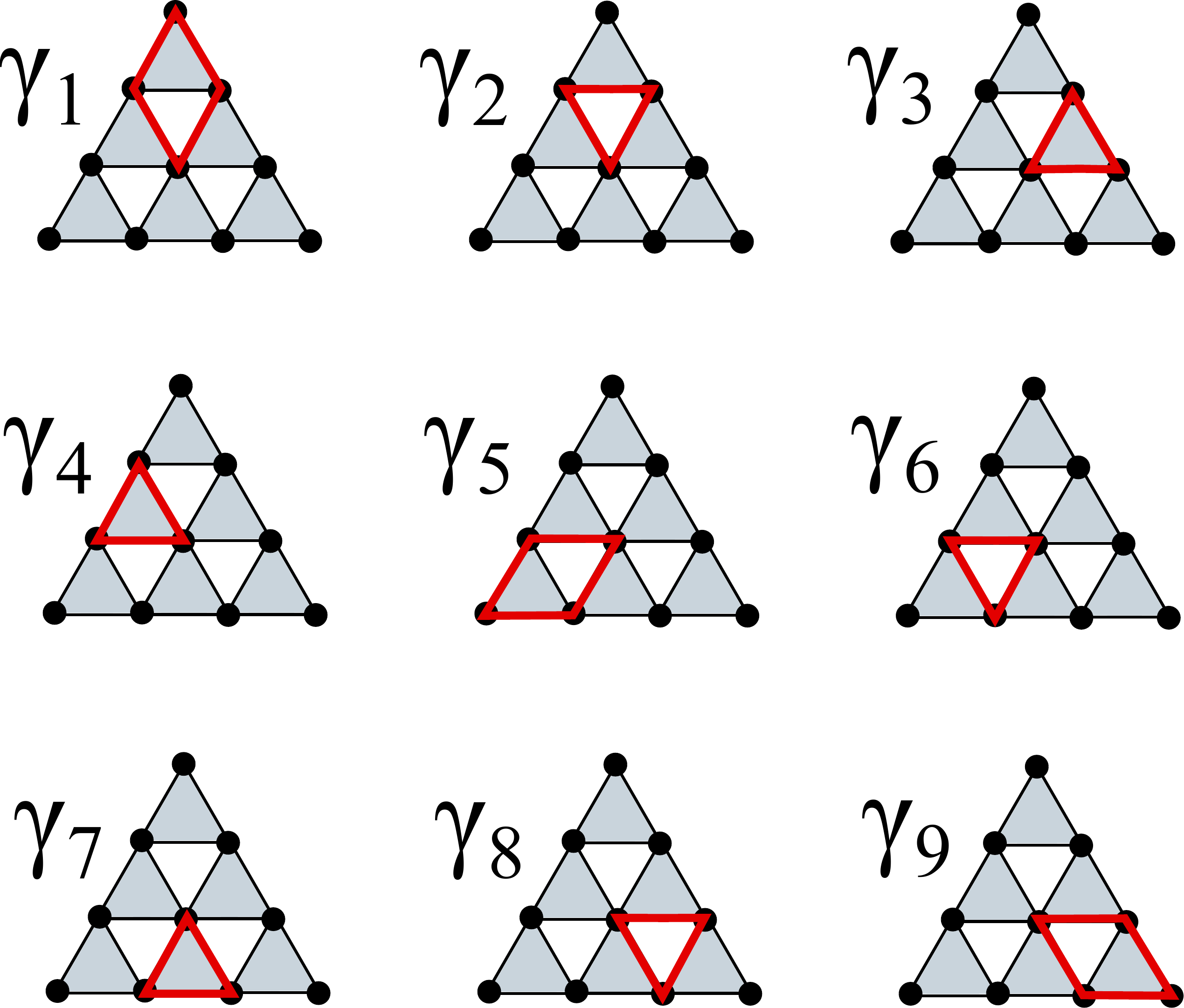}
	\hspace{1cm}
	{\bf b)}\qquad\includegraphics[width = .35\textwidth]{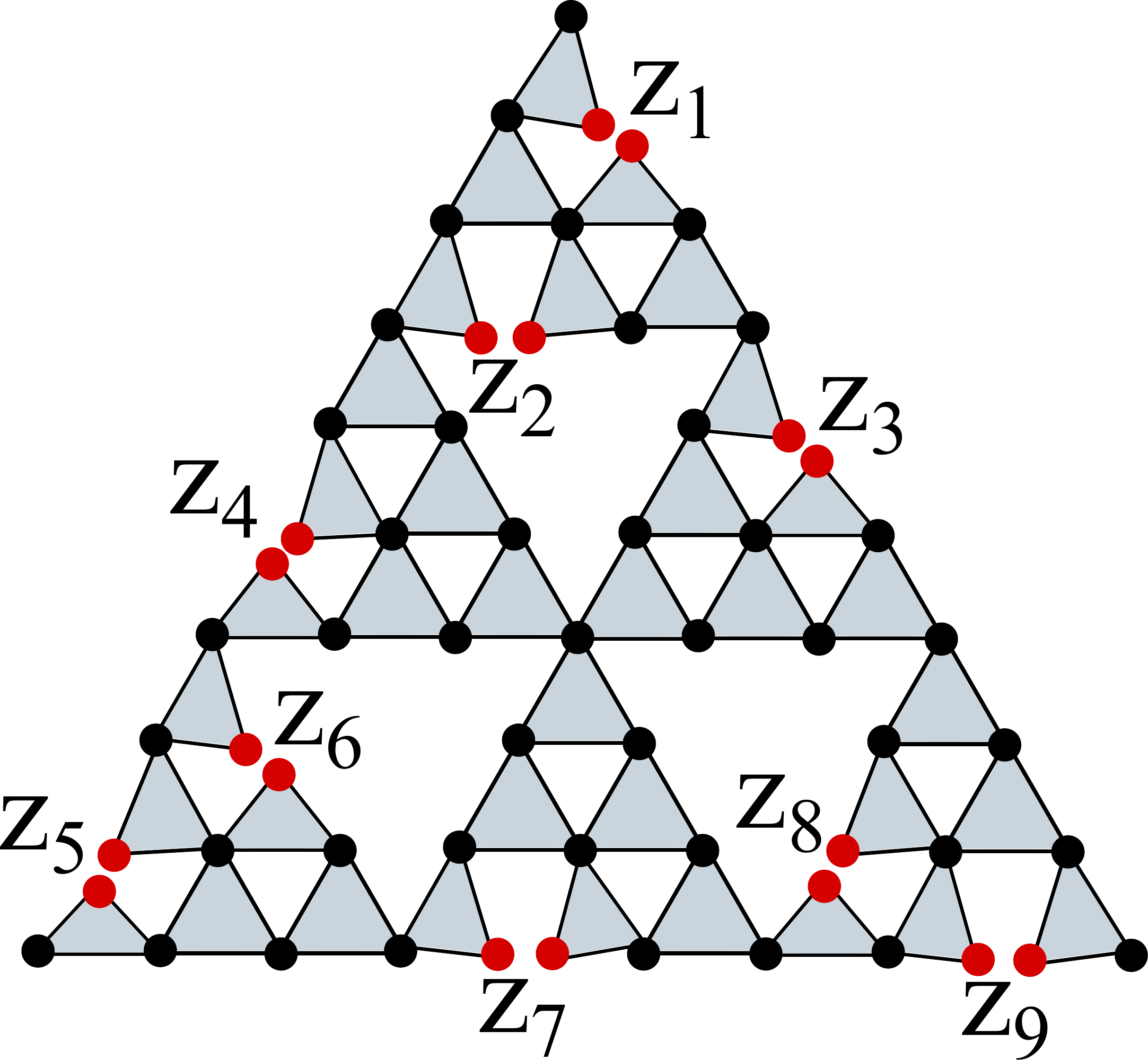}
	\caption{Covering space construction for the higher-order Sierpinski Gasket, $SG_3$. {\bf a)} A basis for the cycle space of $\Gamma_1$. In particular, the boundaries of the $n$-cells are not sufficient to form a basis. {\bf b)} Associated cut vertices on $\Gamma_2$. Superscripts omitted for simplicity.}
	\label{fig:sg3_cut}
\end{figure}

\subsection{Convergence of KM on p.c.f. fractals}

An analogous theory of harmonic functions exists on these covering spaces. Introduce Dirichlet energies on the fundamental domain,
\begin{equation*}
	\cE_m[\Gamma_m^0](u) = \sum_{\{i,j\}\in E^0_m} c_{ij}^m\frac{(u(v_j)-u(v_i))^2}{2},
\end{equation*}
and corresponding spaces
\begin{align*}
	H^0_m = \{ f \in L(V^0_m,\R)\colon  &f(v^0_1)=0,\\
	&\; f((z^0_i)_+)=f((z^0_i)_-)+\rho_i, \forall i\in [n_m]\},
\end{align*}
where, as usual, the value at $v_1^0\in V^0$ is fixed to remove translation invariance. As in \eqref{eq:limit_k0_energy}, the limit
\begin{equation*}
	\cE_{K^0}(f) = \lim_{m\to\infty} \cE_m[\Gamma_m^0](f)
\end{equation*}
is well-defined. Denote $\operatorname{dom}(\cE_{K^0})\doteq \{u\in C(K^0,\R) : \cE_{K^0}(u)<\infty\}$.

Minimizing $\cE_m[\Gamma_m^0]$ and repeatedly applying the harmonic extension algorithm results in a harmonic function on $K^0$, which can be extended to all of $\tilde{K}$ as in Section \ref{sec.harm-struct}. Projecting the range to $\T$ results in the desired harmonic map. Moreover, the resultant map satisfies natural boundary conditions (cf. Lemma \ref{lem.natural}).

\begin{theorem}\label{thm.pcf}
	Fix $m\in \N$ and a degree
	\begin{equation}\label{eq:winding_vec_general}
		\bar{\omega}^{(m)}(f)=(\rho_1,\dots,\rho_{n_m}).
	\end{equation}
	There is a harmonic map $f\colon K\to\mathbb{T}$ satisfying \eqref{eq:winding_vec_general} and $\partial_{\bf n}f(v_i)=0$ for all $v_i\in V^0$.  \end{theorem}

As in the previous sections, to establish convergence results we require H\"older continuity. A variation of the following can be found in \cite{Str06}. Recall that the fundamental domain $K^0$ can be represented as a finite union of $m$-cells. Denoting $K^0_w\doteq K_w(K^0)$, we have $K^0 = \bigcup_{|w|=m}K^0_w$, where each $K^0_w$ is a copy of $K$.

\begin{lemma}\label{lem:holder_estimate_pcf}
	Let $f\in \operatorname{dom}(\mathcal{E}_{K^0})$, and represent the fundamental domain as a union of $m$-cells: $K^0=\bigcup_{|w|=m} K^0_w$. There exists $\beta>0$ so that for any $x,y \in K^0_w$,  
	\begin{equation*}
		\frac{|f(x)-f(y)|}{|x-y|^{\beta}}\leq C\sqrt{\mathcal{E}_{K^0}(f)},
	\end{equation*}
	where $C$ is independent of $f$.
\end{lemma}

\begin{proof}
	Arguing as in Lemma \ref{lem:holder_sg_lift}, it is sufficient to establish H\"older continuity on $K$ itself. Let $R(x,y)$ be the effective resistance metric induced by $\cE_{K}$ on $K$. That is,
	\begin{equation*}
		R(x,y)^{-1}\doteq \min \{\cE(u)\colon u(x)=0,\; u(y)=1\}.
	\end{equation*} 
	Then $R$ is a metric and induces a 1/2-H\"older norm on the domain of $\cE_{K}$, cf \cite{Kig01}:
	$$|f(x)-f(y)|\leq R(x,y)^{1/2}\sqrt{\cE(f)}.$$
	Hence, we need only show that $R(x,y)^{1/2}\leq C|x-y|^\beta$. 
	
	To that end, assume first that $x,y\in V_n$ and take $u\in \operatorname{dom}(\cE)$ to be the harmonic spline whose values on $V_n$ are $u(z)=\delta_{zy}$ (i.e., $u(z)=1$ only when $z=y$). Then $\cE(u)\geq \cE_n(u)\geq Cr_w^{-1}$, where $r_w$ is a product of the $r_i$ defined in \eqref{eq:compatability0}. From the definition of $R$ certainly $R(x,y)\leq Cr_w$. Defining $\bar{r}=\max r_i <1$, we have $R(x,y)\leq C\bar{r}^n$. Next, recall that the $F_i$ have contraction ratios $\lambda_i$. Letting $\underline{\lambda}=\min \lambda_i$, we have a lower bound $|x-y|\geq C\underline{\lambda}^n$, $x,y \in V_n$. Thus taking $\alpha = \log(\overline{r})/\log(\underline{\lambda})$ then 
	\begin{equation*}
		R(x,y)^{1/2}\leq C |x-y|^{\alpha/2}.
	\end{equation*}
	If $x,y\in K$ are arbitrary, then by density of $V_*$ there are sequences $x_n\to x$ and $y_n\to y$ with $x_n,y_n\in V_n$ and the same estimate follows by continuity.
\end{proof}

The convergence of stable equilibria of KM on p.c.f. fractals now follows similarly. Consider the KM on $\Gamma_m$ approximating $K$:
\begin{equation}\label{KM-pcf}
	\dot u(t,v_i)= \sum_{j:\,j\sim_m i}c_{i,j}^m \sin\left(2\pi\left(u(t,v_j)-u(t,v_i)\right)\right), \quad
	v_i\in  V_m.
\end{equation} 
Combining the covering space construction of harmonic maps (Theorem \ref{thm.pcf}), and the a priori H\"older estimates on the fundamental domain (Lemma \ref{lem:holder_estimate_pcf}), applying the same arguments as in Section \ref{sec.proof} results in the following Theorem.

\begin{theorem}
	Fix $m\in \N$ and a degree $\bar{\omega}^{(m)}(u^*)$. Construct a harmonic map $u^*\in C(K,\T)$ as in Theorem \ref{thm.pcf}. Then, for each $\varepsilon>0$, there exists $N\in \N$ such that for all $n\geq N$, the KM on $\Gamma_n$ has a stable steady state solution $u^{n}\in L(V_n,\T)$, which can be extended to a continuous function on $K$, $\tilde{u}^n$, such that
	$\bar{\omega}^{(m)}(\tilde u^{n})=\bar{\omega}^{(m)}(u^*)$ and
	\begin{equation}
		\max_{x\in K} | \tilde{u}^{n}(x)-u^*(x)|<\epsilon.
	\end{equation}
\end{theorem}


\noindent\textbf{Author Contributions}\; 
G.S.M. and M.S.M. contributed to writing the manuscript equally. All authors reviewed the manuscript.

\noindent\textbf{Funding}\;
The work of G.S.M and M.S.M. was partially supported by NSF DMS Awards \# 2406941  and \\
\# 2406942 respectively.

\noindent\textbf{Data availability}\; No datasets were generated or analysed during the current study.

\bibliographystyle{amsplain}
\bibliography{gamma}

\end{document}